\documentclass[10pt,journal,compsoc]{IEEEtran}
\IEEEoverridecommandlockouts       

\usepackage{epsfig} 
\usepackage{amsmath} 
\usepackage{amssymb}  
\usepackage{amsthm}

\usepackage{cite}      
\usepackage{graphicx}
\usepackage{caption}
\usepackage{subcaption}
\usepackage{url}       

\usepackage{verbatim}   
\usepackage{color}
\usepackage{tabulary}
\usepackage{booktabs}

\newcommand{\HIGH}[1]{{\color{black}{#1}}}
\newcommand{\high}[1]{{\color{black}{#1}}}

\usepackage{algorithm}
\usepackage{algorithmic} 

\usepackage{bbold}
\usepackage{multirow}

\begin{document}

\title{
Node-based Service-Balanced Scheduling 
for Provably Guaranteed Throughput and Evacuation Time Performance
}


\author{Yu~Sang,
Gagan~R.~Gupta, 
and Bo~Ji~\IEEEmembership{Member,~IEEE} 
\thanks{This work was supported by the NSF under Grant CNS-1651947.
\HIGH{A preliminary version of this work was presented at the IEEE INFOCOM,
San Francisco, CA, April 10 -- 15, 2016 \cite{ji16infocom}.}}
\thanks{Yu Sang (yu.sang@temple.edu) and Bo Ji (boji@temple.edu) are with Dept. of CIS at Temple University. 
Gagan R. Gupta (gagan.gupta@iitdalumni.com) is with AT\&T Labs.
Bo Ji is the corresponding author.}
}

\newtheorem{theorem}{\it Theorem}
\newtheorem{proposition}{\it Proposition}
\newtheorem{lemma}{\it Lemma}
\newtheorem{corollary}{\it Corollary}
\newtheorem{definition}{\it Definition}
\newtheorem{conjecture}{\it Conjecture}

\newcommand{\ChromaticIndex}{\mathcal{X}^{\prime}}
\newcommand{\Graph}{G}
\newcommand{\Vertex}{V}
\newcommand{\Edge}{E}
\newcommand{\Matching}{\mathcal{M}}

\newcommand{\Odd}{\mathcal{O}}
\newcommand{\xr}{{x_r}}
\newcommand{\xrj}{{x_{r_j}}}
\newcommand{\DIFF}{{\textstyle \frac{D^+}{dt^+}}}
\newcommand{\diff}{{\textstyle \frac{d}{dt}}}
\newcommand{\Expect}{\operatorname{E}}
\newcommand{\Var}{\operatorname{Var}}
\newcommand{\mT}{\mathcal{T}}
\newcommand{\mC}{\mathcal{C}}
\newcommand{\mH}{\mathcal{H}}

\newcommand{\argmax}{\operatornamewithlimits{argmax}}

\DeclareGraphicsExtensions{.pdf,.eps}

\IEEEtitleabstractindextext{%
\begin{abstract}
This paper focuses on the design of \emph{provably efficient online} link scheduling algorithms for multi-hop 
wireless networks. We consider single-hop traffic and the one-hop interference model. The objective is twofold: 
1) \emph{maximizing the throughput} when the flow sources continuously inject packets into the network, and
2) \emph{minimizing the evacuation time} when there are no future packet arrivals. 
The prior work mostly employs the link-based approach, which leads to throughput-efficient algorithms 
but often does not guarantee satisfactory evacuation time performance. In this paper, 
\HIGH{we propose a novel Node-based Service-Balanced (NSB) online scheduling algorithm.
NSB aims to give scheduling opportunities to heavily congested nodes in a balanced manner, by maximizing
the total weight of the scheduled nodes in each scheduling cycle, where the weight of a node is 
determined by its workload and whether the node was scheduled in the previous scheduling cycle(s).}
We rigorously prove that NSB guarantees to achieve an efficiency 
ratio no worse (or no smaller) than $2/3$ for the throughput and an approximation ratio no worse 
(or no greater) than $3/2$ for the evacuation time. It is remarkable that NSB is both throughput-optimal 
and evacuation-time-optimal if the underlying network graph is bipartite. Further, we develop a lower-complexity 
NSB algorithm, called LC-NSB, which provides the same performance guarantees as NSB.
Finally, we conduct numerical experiments to elucidate our theoretical results.

\end{abstract}

}

\maketitle
\section{Introduction} \label{sec:intro}
Resource allocation is an important problem in wireless networks. Various functionalities
at different layers (transport, network, MAC, and PHY) need to be carefully designed so as to 
efficiently allocate network resources and achieve optimal or near-optimal network performance. 
Among these critical functionalities, link scheduling at the MAC layer, which, at each time decides 
which subset of non-interfering links can transmit data, is perhaps the most challenging component
and has attracted a great deal of research effort in the past decades (see \cite{georgiadis06,lin06c} 
and references therein).

In this paper, we focus on the design of \emph{provably efficient online} link scheduling algorithms 
for multi-hop wireless networks with \emph{single-hop} traffic under the \emph{one-hop} 
interference model\footnote{The packets of single-hop traffic traverse only one link before leaving 
the system. \HIGH{The one-hop interference model is also called the node-exclusive or the primary 
interference model, where two links sharing a common node cannot be active at the same time.
This model can properly represent practical wireless networks based on Bluetooth or FH-CDMA 
technologies \cite{hajek88,sarkar05,lin06,lin06c,joo09}.}}.
While throughput is widely shared as the first-order performance metric, which characterizes the long-term 
average traffic load that can be supported by the network, 
\HIGH{evacuation time is also of critical importance due to the following reasons.
First, draining all existing packets within a minimum amount of time is a major concern in the settings 
without future arrivals.
One practical example is environmental monitoring using wireless sensor networks, 
where all measurement data periodically generated by different nodes at the same time, 
need to be transmitted to one or multiple sinks for further processing.
Second, evacuation time is also highly correlated with the delay performance in the settings with arrivals. 
For example, evacuation-time optimality is a necessary condition for the strongest delay notion 
of sample-path optimality \cite{guptathesis}.
Third, it is quite relevant to timely transmission of delay-sensitive data traffic (e.g., deadline-constrained 
packet delivery) \cite{soldati09,singh16}.}

However, these different metrics may lead to conflicting scheduling decisions -- an algorithm designed 
for optimizing one metric may be detrimental to the other metric (see \cite{guptathesis} for such examples). 
Therefore, it is challenging to design an efficient scheduling algorithm that can provide provably guaranteed 
performance for both metrics at the same time. 

While throughput has been extensively studied since the seminal work by Tassiulas and Ephremides 
\cite{tassiulas92} and is now well understood, evacuation time is much less studied. In the no-arrival 
setting, the minimum evacuation time problem \high{is equivalent} to the \emph{multigraph\footnote{In 
a multigraph, more than one edge, called multi-edge, is allowed between two nodes.} edge coloring problem}
\HIGH{due to the following: each multi-edge corresponds to a packet waiting to be transmitted over the link 
between the nodes of the multi-edge; each color corresponds to a feasible schedule (or a matching); finding 
the chromatic index (i.e., the minimum number of colors such that, each multi-edge is assigned a color and 
two multi-edges sharing a common node cannot have the same color) is equivalent to minimizing the time 
for evacuating all the packets by finding a matching at a time.}
Since edge-coloring is a classic NP-hard problem \cite{holyer81}, a rich body of research has focused on 
developing approximation algorithms (see \cite{stiebitz12} for a good survey). 
These algorithms employ a popular recoloring technique that requires computing the colors all at once, 
and yield a complexity that depends on the number of multi-edges. This, however, renders them 
unsuitable for application in a network with arrivals. This is because the complexity would become 
impractically high when there are a large number of packets (or multi-edges) in the network.
Therefore, it is desirable to have an \emph{online} scheduling algorithm that at each time quickly 
computes one schedule (or color) based on the current network state (e.g., the queue lengths) and 
yields a complexity that only depends on the node count $n$ and/or the link count $m$.

Most existing online scheduling algorithms either make scheduling decisions based on the link load (such as 
Maximum Weighted Matching (MWM) \cite{tassiulas92} and Greedy Maximal Matching (GMM) \cite{lin06,joo09}) 
or are load agnostic (such as Maximal Matching (MM) \cite{lin06,wu07}). While these algorithms are throughput-efficient, 
none of them can guarantee an approximation ratio better (or smaller) than $2$ for the evacuation time \cite{guptathesis}. 
In contrast, several prior work \cite{mekkittikul98,tabatabaee09,guptathesis,ji15ciss} proposes algorithms based 
on the node workload (i.e., packets to transmit or receive), such as the Lazy Heaviest Port First (LHPF) 
algorithms, which are both throughput-optimal and evacuation-time-optimal in input-queued switches (which 
can be described as bipartite graphs) \cite{guptathesis}. The \emph{key intuition} behind the node-based 
approach is that the minimum evacuation time is lower bounded by the largest workload at the nodes and 
the odd-size cycles, and this lower bound is asymptotically tight \cite{kahn96}.
\high{Hence, giving a higher priority to scheduling nodes with heavy workload leads to better evacuation time performance,
while the link-based approach that fails to respect this crucial fact results in unsatisfactory evacuation time performance.}

While the node-based approach seems quite promising, the scheduling performance of the node-based algorithms 
is not well understood, and the existing studies are mostly limited to bipartite graphs \cite{mekkittikul98,tabatabaee09,guptathesis}. 
Very recent work of \cite{ji15ciss} considers general network graphs and shows that the Maximum Vertex-weighted 
Matching (MVM) algorithm can guarantee an approximation ratio no worse (or no greater) than $3/2$ for the evacuation 
time. However, throughput performance of MVM remains unknown. 

There is several other related work. In \cite{georgiadis15}, the authors study the connection between 
throughput and (expected) minimum evacuation time, but no algorithms with provable performance 
guarantees are provided. The work of \cite{tassiulas95,tassiulas99,soldati09} considers the minimum 
evacuation time problem for multi-hop traffic in some special scenarios (e.g., special network topologies 
or wireline networks without interference).

In this paper, \emph{the goal is to develop efficient online link scheduling algorithms that can provide 
provably guaranteed performance for both throughput and evacuation time.} We summarize our 
contributions as follows.

First, we propose a Node-based Service-Balanced (NSB) scheduling algorithm that makes scheduling 
decisions based on the node workload and whether the node was scheduled in the previous time-slot(s). 
NSB has a complexity of $O(m \sqrt{n} \log{n})$.
We rigorously prove that NSB guarantees to achieve an approximation ratio no worse (or no greater) 
than $3/2$ for the evacuation time and an efficiency ratio no worse (or no smaller) than $2/3$ for the 
throughput. It is remarkable that NSB is both throughput-optimal and evacuation-time-optimal if the underlying 
network graph is bipartite. The \emph{key novelty} of NSB is that it takes a node-based approach and 
gives balanced scheduling opportunities to the bottleneck nodes with heavy workload. A novel application 
of \emph{graph-factor theory} is adopted to analyze how NSB schedules the heavy nodes (Lemma~\ref{lem:existence}).

Second, from the performance analysis for NSB, we learn that in order to achieve the same performance guarantees, 
what really matters is the priority or the ranking of the nodes, rather than the exact weight of the nodes. 
Using this insight, we develop the Lower-Complexity NSB (LC-NSB) algorithm. We show that LC-NSB 
can provide the same performance guarantees as NSB, while enjoying a lower complexity of $O(m \sqrt{n})$.

In Table~\ref{tab:com}, we summarize the guaranteed performance of NSB and LC-NSB as well as 
several most relevant online algorithms in the literature. 
\emph{As can be seen, none of the existing algorithms strike a more balanced performance
guarantees than NSB and LC-NSB in both dimensions of throughput and evacuation time.}
%
Finally, we conduct numerical experiments to validate our theoretical results and compare the empirical
performance of various algorithms.


\begin{table}
\scriptsize
\centering
\begin{tabular}{|c|c|c|c|c|c|}\hline
  \multirow{2}{*}{Algorithm} & \multirow{2}{*}{Complexity} & \multicolumn{2}{|c|}{$\gamma$ (Throughput)} & \multicolumn{2}{|c|}{$\eta$ (Evacuation time)}  \\ \cline{3-6}
  & & General & Bipartite & General & Bipartite \\ \hline
  MWM & $O(mn)$ & 1 & 1 & 2 & 2 \\ \hline
  GMM & $O(m \log{m})$ & $\ge 1/2$ & $\ge 1/2$ & 2 & 2\\ \hline
  MM & $O(m)$ & $\ge 1/2$ & $\ge 1/2$ & 2 & 2 \\ \hline
	 MVM & $O(m \sqrt{n} \log{n})$ & ? & 1 & $\le 3/2$ & 1 \\ \hline
  \textbf{NSB} & $O(m \sqrt{n} \log{n})$ & $\ge 2/3$ & 1 & $\le 3/2$ & 1 \\ \hline
  \textbf{LC-NSB} & $O(m \sqrt{n})$ & $\ge 2/3$ & 1 & $\le 3/2$ & 1 \\ \hline
\end{tabular}
\caption{Performance comparison of NSB and LC-NSB with several most relevant online algorithms in the literature.
	     The efficiency ratio $\gamma$ and the approximation ratio $\eta$ are used for comparing the performance 
	     of throughput and evacuation time, respectively. (See formal definitions of $\gamma$ and $\eta$ in 
	     Section~\ref{sec:model}.) For both $\gamma$ and $\eta$, a value closer to 1 is better.
	     The complexity provided here is for making a scheduling decision at each time.
}
\label{tab:com}
\end{table}


The remainder of this paper is organized as follows. First, we describe the system model and 
the performance metrics in Section~\ref{sec:model}. Then, we propose the NSB algorithm and 
analyze its performance in Section~\ref{sec:nsb}. A lower-complexity NSB algorithm with the 
same performance guarantees is developed in Section~\ref{sec:lc-nsb}. Finally, we conduct 
numerical experiments in Section~\ref{sec:sim} and make concluding remarks in Section~\ref{sec:con}.
\high{Some detailed proofs are provided in Section~\ref{sec:proofs}.}

\section{System Model} \label{sec:model}

We consider a multi-hop wireless network described as an undirected graph $\Graph=(\Vertex,\Edge)$, 
where $\Vertex$ denotes the set of nodes and $\Edge$ denotes the set of links. The node count 
and the link count are denoted by $n = \lvert \Vertex \rvert$ and $m = \lvert \Edge \rvert$, respectively. 
Nodes are wireless transmitters/receivers and links are wireless channels between two nodes. 
The set of links touching node 
$i \in \Vertex$ is defined as $L(i) \triangleq \{l \in \Edge ~|~ i~\text{is an end node of link}~l\}$. 
We assume a time-slotted system with a single frequency channel. We also assume unit link 
capacities, i.e., a link can transmit at most one packet in each time-slot when active. However, 
our analysis can be extended to the general scenario with heterogeneous link capacities by 
considering the workload defined as $\lceil \text{number of packets} / \text{link capacity} \rceil$. 
We consider the \emph{one-hop} interference model, under which a feasible schedule corresponds 
to a \emph{matching} (i.e., a subset of links, $L$, that satisfies that no two links in $L$ share a common node). 
A matching is called \emph{maximal}, if no more links can be added to the matching without 
violating the interference constraint. We let $\Matching$ denote the set of all matchings over $\Graph$. 

\HIGH{As in several previous work (e.g., \cite{lin06,joo09c,leconte11,ni12,li17}), 
we focus on link scheduling at the MAC layer, and thus we only consider single-hop traffic.}
We let $A_l(k)$ denote the cumulative amount of workload (or packet) arrivals at 
link $l \in \Edge$ up to time-slot $k$ (including time-slot $k$). By slightly abusing the notations,
we let $A_i(k) \triangleq \sum_{l \in L(i)} A_l(k)$ denote the cumulative amount of workload arrivals 
at node $i \in \Vertex$ up to time-slot $k$ (including time-slot $k$). (Indices $l$ and $i$ correspond 
to links and nodes, respectively; similar for other notations.) We assume that the arrival process 
$\{A_l(k), k \ge 0\}$ satisfies the strong law of large numbers (SLLN): with probability one,
\begin{equation}
\label{eq:slln_l}
\lim_{k \rightarrow \infty} A_l(k)/k = \lambda_l
\end{equation}
for all links $l \in \Edge$, where $\lambda_l$ is the mean arrival rate of link $l$.
Let $\lambda \triangleq [\lambda_l: l \in \Edge]$ denote the arrival rate vector. 
We assume that the arrival processes are independent across links. Note that 
the process $\{A_i(k), k \ge 0\}$ also satisfies SLLN: with probability one,
$\lim_{k \rightarrow \infty} A_i(k)/k = \lambda_i$
for all nodes $i \in \Vertex$, where $\lambda_i \triangleq \sum_{l \in L(i)} \lambda_l$ is 
the mean arrival rate for node $i$. 

Let $Q_l(k)$ be the queue length of link $l$ in time-slot $k$, and let $D_l(k)$ be the 
cumulative number of packet departures at link $l$ up to time-slot $k$. We assume that
there are a finite number of initial packets in the network at the beginning of time-slot 
0. Let $Q_i(k) \triangleq \sum_{l \in L(i)} Q_l(k)$ be the amount of workload at node 
$i \in \Vertex$ (i.e., the number of packets waiting to be transmitted to or from node 
$i$) in time-slot $k$, and let $D_i(k) \triangleq \sum_{l \in L(i)} D_l(k)$ be the amount 
of cumulative workload served at node $i \in \Vertex$ up to time-slot $k$. We also call 
$Q_i(k)$ and $D_i(k)$ as the queue length and the cumulative departures at node $i$ 
in time-slot $k$, respectively. 

Without loss of generality, we assume that only links with a non-zero queue length can be 
activated. Let $M_l=1$ if matching $M \in \Matching$ contains link $l$, and $M_l = 0$ 
otherwise. Let $H_M(k)$ be the number of time-slots in which $M$ is selected 
as a schedule up to time-slot $k$. We set by convention that $A_i(k)=0$ and $D_i(k)=0$
for all $i \in \Vertex$ and for all $k \le 0$. The queueing equations of the system are as follows:
\begin{flalign}
& Q_i(k) = Q_i(0) + A_i(k) - D_i(k-1), \label{eq:qik}\\
& D_i(k) = \sum_{M \in \Matching} \sum_{\tau=1}^{k} \sum_{l \in L(i)} M_l \cdot (H_M(\tau) - H_M(\tau-1)), \\
& \sum_{M \in \Matching} H_M(k) = k. 
\end{flalign}

Next, we define system stability as follows.
\begin{definition}
\label{def:stab}
The network is rate stable if with probability one,
\begin{equation}
\label{eq:stab}
\lim_{k \rightarrow \infty} D_l(k)/k = \lambda_l,
\end{equation}
for all $l \in \Edge$ and for any arrival processes satisfying Eq.~(\ref{eq:slln_l}).
\end{definition}

Note that we consider \emph{rate stability} for ease of presenting our main ideas.
Strong stability can similarly be derived if we make stronger assumptions on the 
arrival processes \cite{andrews04}.

We define the \emph{throughput region} of a scheduling algorithm as the set of arrival rate
vectors for which the network remains rate stable under this algorithm. Further, we define the 
\emph{optimal throughput region}, denoted by $\Lambda^*$, as the union of the throughput 
regions of all possible scheduling algorithms. A scheduling algorithm is said to have an 
\emph{efficiency ratio} $\gamma$ if it can support any arrival rate vector $\lambda$ strictly 
inside $\gamma \Lambda^*$. Clearly, we have $\gamma \in [0,1]$. In particular, a scheduling 
algorithm with an efficiency ratio $\gamma=1$ is \emph{throughput-optimal}, i.e., it can stabilize 
the network under any feasible load. We also define another important region $\Psi$ by considering
bottlenecks formed by the nodes:
\begin{equation}
\label{eq:psi}
\Psi \triangleq \{\lambda ~|~ \lambda_i \le 1 ~\text{for all}~ i \in \Vertex \}.
\end{equation}
Clearly, we have $\Lambda^* \subseteq \Psi$ because at most one packet can be transmitted 
from or to a node in each time-slot. 
\HIGH{
Similarly, any odd-size cycle $Z$ could also be a bottleneck because at most $(|Z|-1)/2$ 
out of the $|Z|$ links of the odd-size cycle can be scheduled at the same time. 
For example, the total arrival rate summed over all edges of a triangle must not exceed 
$1/3$ because at most one out of the
three links of the triangle can be scheduled in each time-slot.
For the theoretical analyses, we consider bottlenecks formed only by the nodes, which is
sufficient for deriving our analytical results. We provide more discussions about 
odd-size cycles in Sections~\ref{subsec:bipartite} and \ref{subsec:throughput}.
}

As we mentioned earlier, in the settings without future packet arrivals, the performance metric 
of interest is the evacuation time, defined as the time interval needed for draining all the initial 
packets. Let $T^P$ denote the evacuation time of scheduling algorithm $P$, and let 
$\ChromaticIndex$ denote the minimum evacuation time over all possible algorithms. 
A scheduling algorithm is said to have an \emph{approximation ratio} $\eta$ if it has an 
evacuation time no greater than $\eta \ChromaticIndex$ in any network graph with any finite 
number of initial packets. Clearly, we have $\eta \ge 1$. In particular, a scheduling algorithm 
with an approximation ratio $\eta=1$ is \emph{evacuation-time-optimal}.

In this paper, the goal is to develop efficient online link scheduling algorithms that can 
simultaneously provide provably good performance in both dimensions of throughput and 
evacuation time, measured through the efficiency ratio (i.e., the larger the value of $\gamma$, 
the better) and the approximation ratio (i.e., the smaller the value of $\eta$, the better), respectively. 
\HIGH{Note that the throughput performance has been extensively studied under quite general models
in the literature (see \cite{lin06c,georgiadis06} and references therein), where multi-hop traffic, 
general interference models, and time-varying channels (which can model mobility and fading) 
have been considered. However, the evacuation time performance is much less understood. 
As we have mentioned earlier, even in the setting we consider (assuming single-hop traffic, 
the one-hop interference model, and fixed link capacities), the minimum evacuation time problem 
is already very challenging (i.e., NP-hard). 
Considering multi-hop traffic adds another layer of difficulty. 
This is mainly because of the dependence between the upstream and downstream queues, 
since the arrival process to an intermediate queue is no longer exogenous, but instead, it is 
the departure process of its previous-hop queue. In addition to link scheduling, we also need
to decide which flow's packets will be transmitted when a link is activated. This further complicates
the minimum evacuation time problem.
}

\HIGH{
For quick reference, we summarize the key notations of this paper in Table~\ref{tab:notations}.
}

\begin{table}[!t]
\begin{tabular}{ c p{7cm}}
\toprule 
Symbol  & Meaning  \\
\hline 
 $G$  &  Network topology as an undirected graph       \\
 $V$  &  Set of nodes      \\
 $E$  &  Set of links     \\
 $n$  &  Number of nodes     \\
 $m$ &  Number of links     \\
 $L(i)$   &  Set of links touching node $i \in V$      \\
 $M$   & A matching over $G$ \\
 $\Matching$   & Set of all the matchings over $G$      \\
 $\lambda_l$  &   Mean arrival rate of link $l$   \\
 $\lambda_i$   &  Mean arrival rate of node $i$   \\
 $\Lambda^*$   &   Optimal throughput region    \\
 $\Psi$  &   An outer bound of $\Lambda^*$; see Eq.~\eqref{eq:psi}   \\
 $\gamma$   &  Efficiency ratio (for throughput performance)      \\
 $T^P$   &   Evacuation time of scheduling algorithm $P$    \\
 $ \ChromaticIndex$   &   Minimum evacuation time  \\
 $\eta$   &    Approximation ratio (for evacuation time performance)       \\  
 $A_l(k)$    &    Cumulative arrivals at link $l$ up to time-slot $k$    \\
 $A_i(k)$    &    Cumulative arrivals at node $i$ up to time-slot $k$  \\
 $D_l(k)$   &  Cumulative departures at link $l$ up to time-slot $k$     \\
 $D_i(k)$   &  Cumulative departures at node $i$ up to time-slot $k$   \\
 $Q_l(k)$   &    Queue length at link $l$ in time-slot $k$   \\ 
 $Q_i(k)$   &  Workload at node $i$ in time-slot $k$    \\
 $H_M(k)$   &   Number of time-slots in which matching $M$ is selected as a schedule up to time-slot $k$   \\
 $\Delta(k)$   &  Largest node workload in time-slot $k$\\
 $\mC(k)$      &   Set of critical nodes in time-slot $k$     \\
 $\mH(k)$  	   &   Set of heavy nodes in time-slot $k$      \\
 $R_i(k)$  &  Whether node $i$ is matched in time-slot $k$ or not      \\
 $U_i(k)$  &  See Eq.~(\ref{eq:U})     \\
 $w_i(k)$    &    Weight of node $i$ in time-slot $k$   \\
 $w(M)$ & Weight of matching $M$   \\
\hline
\end{tabular}
\caption{\HIGH{Summary of notations.}}\label{tab:sys}
\label{tab:notations}
\end{table}

\section{Node-based Service-Balanced Algorithm} \label{sec:nsb}
In this section, we propose a novel Node-based Service-Balanced (NSB) scheduling algorithm and 
analyze its performance. Specifically, we prove that NSB guarantees an approximation ratio no worse 
(or no greater) than $3/2$ for the evacuation time (Subsection~\ref{subsec:nsb_et}) and an efficiency 
ratio no worse (or no smaller) than $2/3$ for the throughput (Subsection~\ref{subsec:nsb_throughput}). 
Further, we show that NSB is both throughput-optimal and evacuation-time-optimal in bipartite graphs 
(Subsection~\ref{subsec:bipartite}). To the best of our knowledge, none of the existing algorithms strike 
a more balanced performance guarantees than NSB in both dimensions of throughput and evacuation time.

\subsection{Algorithm} \label{subsec:nsb_alg}

\begin{algorithm}[!t]
\caption{Node-based Service-Balanced (NSB)}
\begin{algorithmic}[1]
\label{alg:nsb}
\STATE In each time-slot $k$:
\FOR{each node $i \in \Vertex$}
\STATE Assign node weight $w_i(k)$ based on Eq.~\eqref{eq:weight}
\ENDFOR
\STATE Exclude links $l \in \Edge$ with $Q_l(k)=0$
\STATE Find an MVM $M^*$ over $\Graph$ with node weight $w_i(k)$'s, i.e., 
\[
M^* \in \argmax_{M \in \Matching} w(M) \triangleq \sum_{i: M \cap L(i) \neq \emptyset} w_i(k)
\]
\FOR{each link $l \in \Edge$}
\IF{$M^*_l=1$}
\STATE Transmit one packet over link $l$
\ELSE
\STATE No transmission over link $l$
\ENDIF
\ENDFOR
\end{algorithmic}
\end{algorithm} 

We start by introducing Maximum Vertex-weighted Matching (MVM)~\cite{spencer84,guptathesis}, 
which will be a key component of the NSB algorithm. Let $w_i$ denote the weight of node $i$. 
We will later describe how to assign the node weights. Also, let $w(M) \triangleq \sum_{i: M \cap L(i) \neq \emptyset} w_i$ 
denote the weight of matching $M$, i.e., the sum of the weight of the nodes matched by $M$. A matching $M^*$ is called 
an MVM if it has the maximum weight among all the matchings, i.e., $M^* \in \argmax_{M \in \Matching} w(M)$. 
In \cite{spencer84}, a very useful property of MVM is proven.
We restate it in Lemma~\ref{lem:path}, which will be frequently used in 
the proofs of our main results.

\begin{lemma}[Lemma~6 of \cite{spencer84}]
\label{lem:path}
For any positive integer $s \le n$, suppose that there exists a matching that matches the $s$ 
heaviest nodes. Then, an MVM matches all of these $s$ nodes too.
\end{lemma}

Now, \high{we consider frames each consisting of three consecutive time-slots and} describe the 
operations of the NSB algorithm. We first give some additional definitions and notations. 
Recall that $Q_i(k)$ denotes the workload of node $i$ in time-slot $k$. 
Let $\Delta(k) \triangleq \max_{i \in \Vertex} Q_i(k)$ denote the largest node queue length in time-slot $k$.
A node $i$ is called \emph{critical} in time-slot $k$ if it has the largest queue length, 
i.e., $Q_i(k) = \Delta(k)$; a node $i$ is called \emph{heavy} 
in time-slot $k$ if its queue length is no smaller than $(n-1)/n \cdot \Delta(k)$. (Our results also hold 
if we replace $(n-1)/n$ with any $\alpha \in [(n-1)/n, 1)$.) It will later become clearer why 
such a threshold is chosen. We use $\mC(k)$ and $\mH(k)$ to denote the set of critical nodes and 
the set of heavy nodes in time-slot $k$, respectively. Let $R_i(k) \triangleq D_i(k)-D_i(k-1)$ denote 
whether node $i$ is matched in time-slot $k$ or not, and define 
\begin{equation}
\label{eq:U}
U_i(k) \triangleq
\begin{cases}
R_i(k-1)R_i(k-2) & \text{if $k=3k^{\prime}+2$}; \\
R_i(k-1) &\text{otherwise},
\end{cases} 
\end{equation}
\high{where $k^{\prime}$ is the frame index. 
Note that $U_i(k)$ is either 1 or 0 and will be used in Eq.~\eqref{eq:weight} to determine whether 
a node needs to get a higher scheduling priority in time-slot $k$ or not. Specifically, the weight of a heavy node $i$ 
is doubled if $U_i(k)=0$ (i.e., this heavy node $i$ did not receive enough service in the previous time-slot(s)) 
such that node $i$ has a higher priority of being scheduled in time-slot $k$.}

\begin{figure*}[t!]
        \centering
        \begin{subfigure}[b]{0.23\linewidth}
                \includegraphics[width=\textwidth]{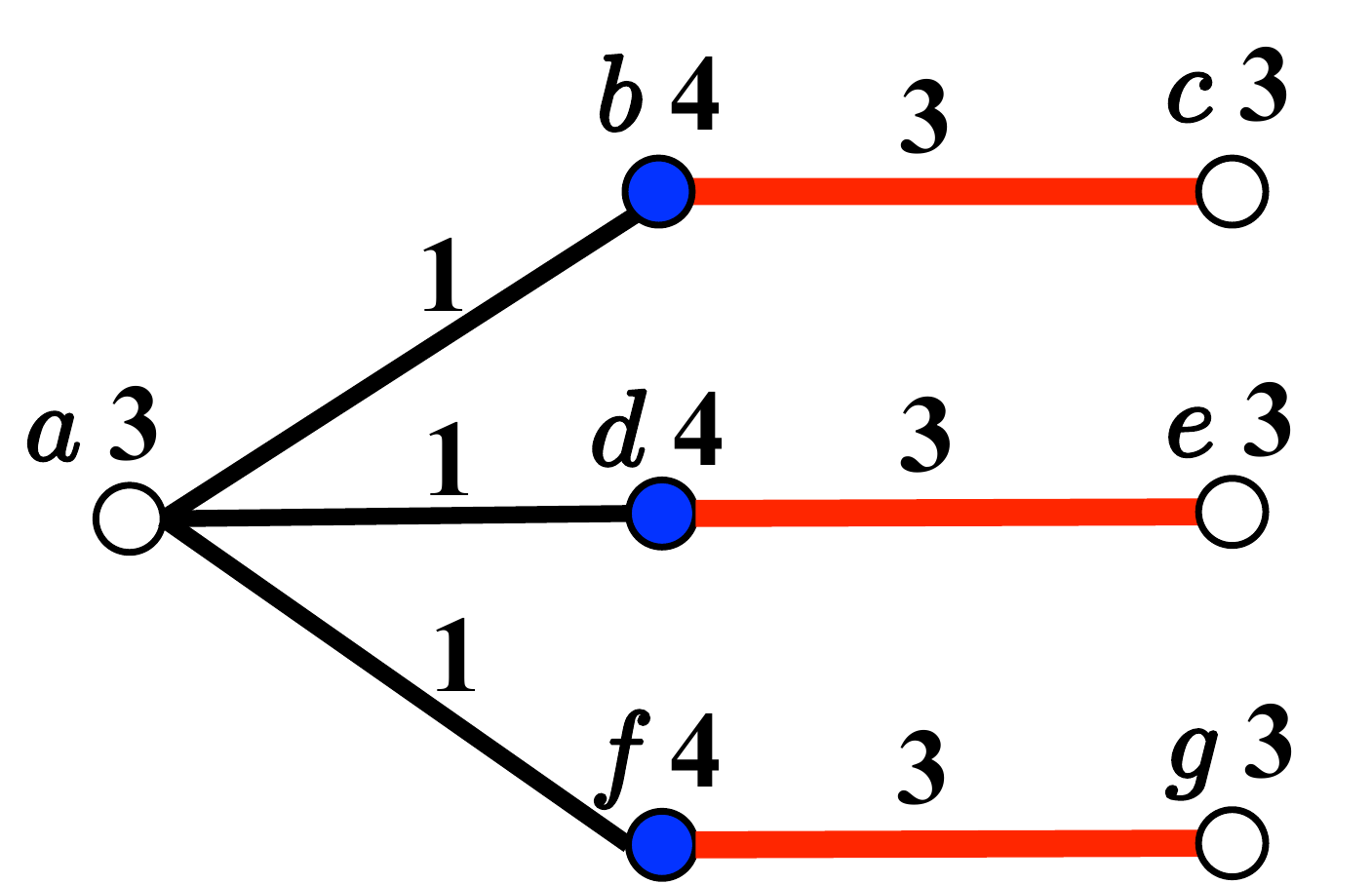}
                \caption{Time-slot 0}
                \label{fig:nsb0}
        \end{subfigure}%
        \begin{subfigure}[b]{0.23\linewidth}
                \includegraphics[width=\textwidth]{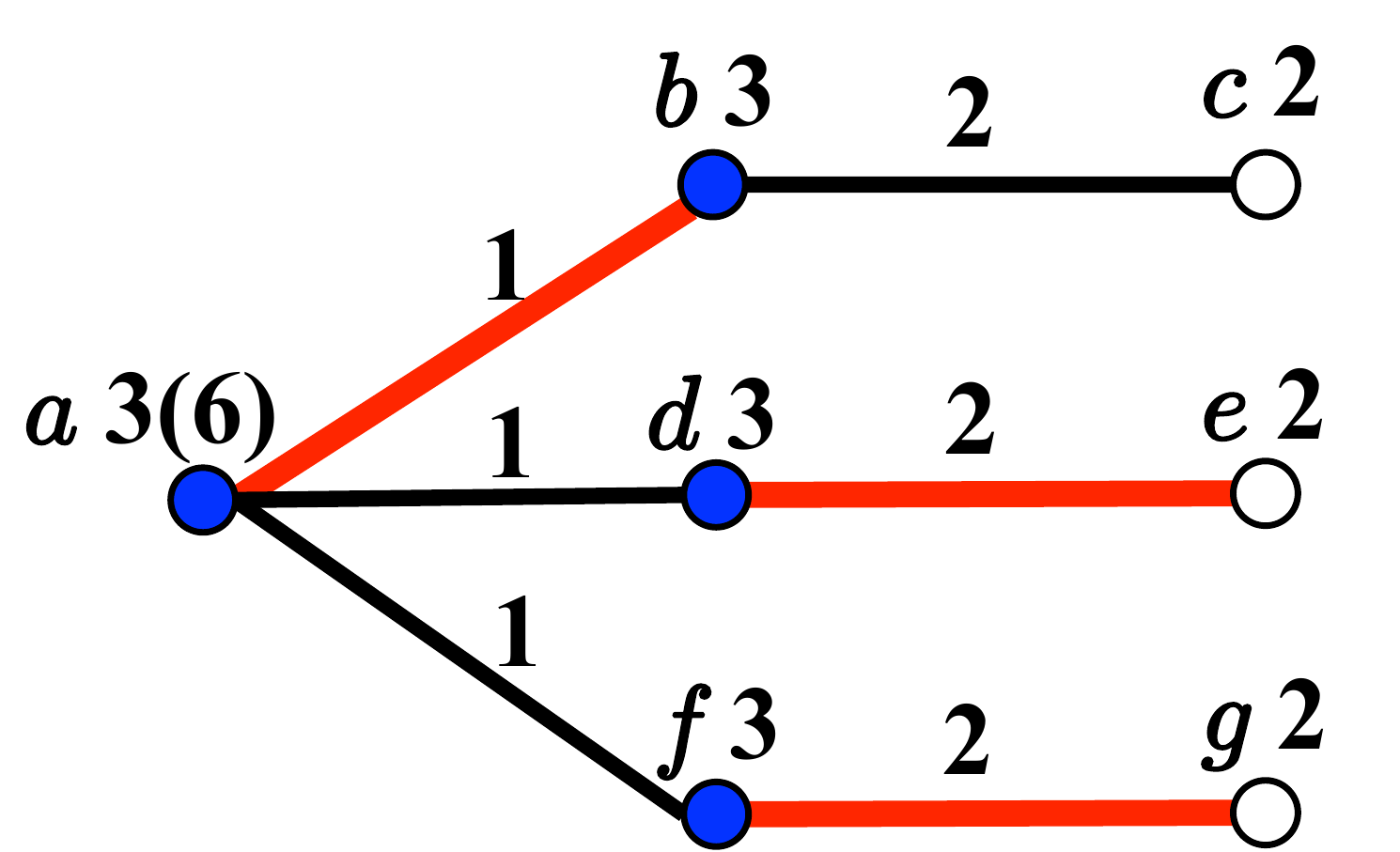}
                \caption{Time-slot 1}
                \label{fig:nsb1}
        \end{subfigure}
	\begin{subfigure}[b]{0.23\linewidth}
                \includegraphics[width=\textwidth]{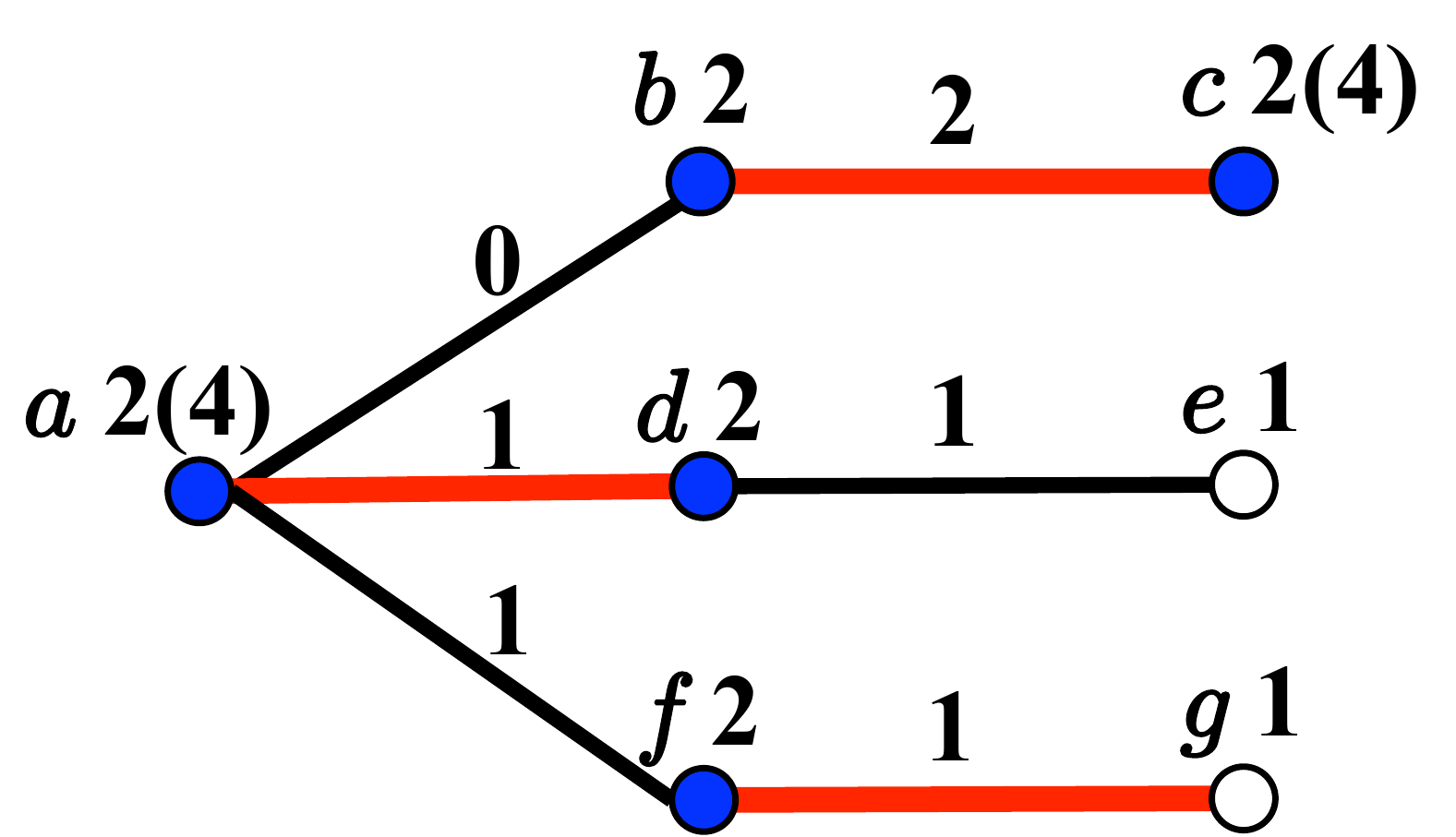}
                \caption{Time-slot 2}
                \label{fig:nsb2}
        \end{subfigure} 
        \begin{subfigure}[b]{0.23\linewidth}
                \includegraphics[width=\textwidth]{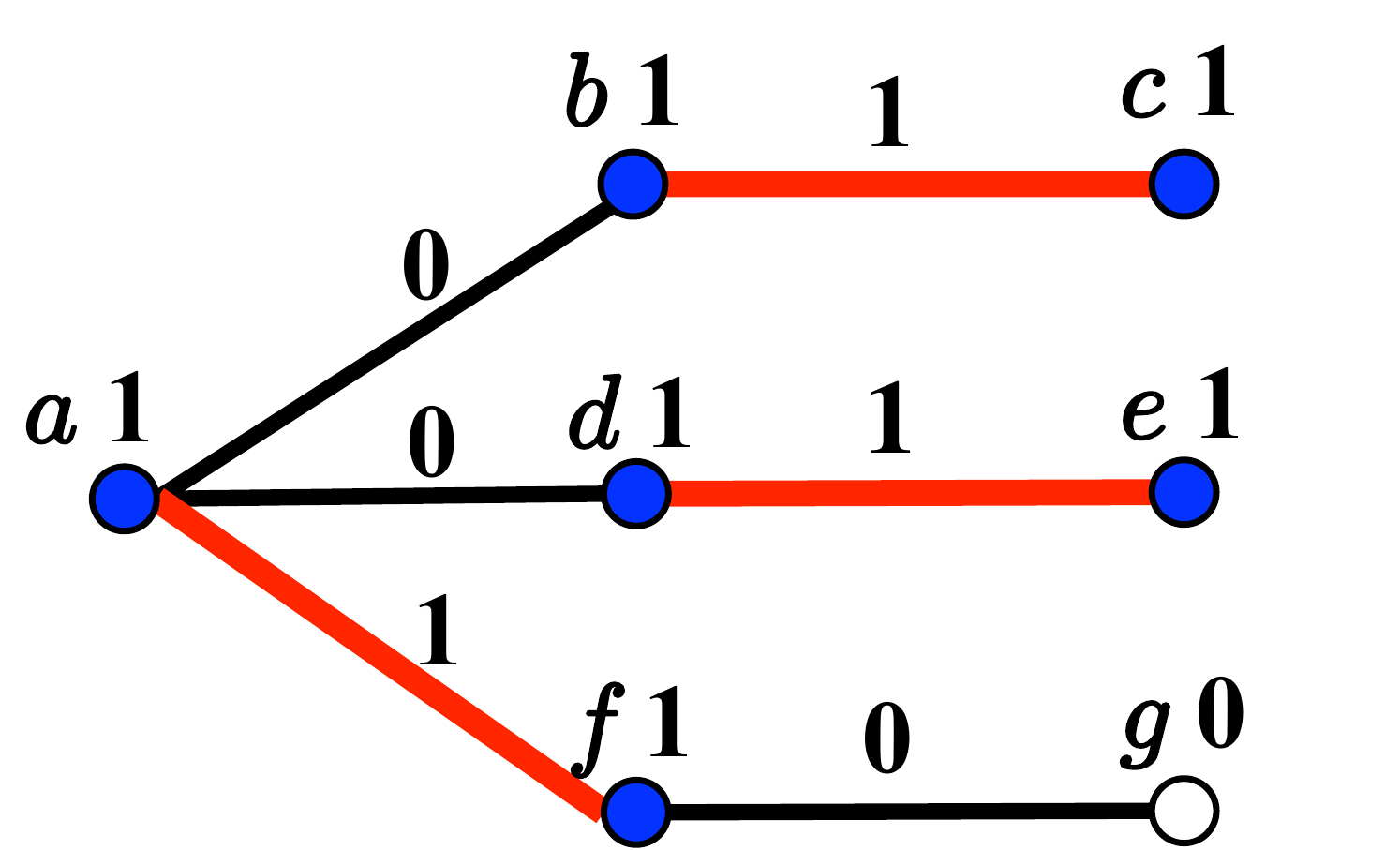}
                \caption{Time-slot 3}
                \label{fig:nsb3}
        \end{subfigure}%
        \caption{\HIGH{An illustration of the operations of NSB in four time-slots. The network setting is presented in Fig.~\ref{fig:nsb0}, 
        			where there are seven nodes $\{a,b,\dots,g\}$. In each subfigure, the number above each link denotes the number of 
			packets waiting to be transmitted over that link at the beginning of each time-slot. 
			For simplicity, we assume no future packet arrivals in this example.
			The node degree (i.e., the sum of queue lengths over all the links touching the node) and the node weight (in the parenthese) 
			are both labeled after the node name; however, the node weight is not labeled if it is equal to the node degree. 
			The heavy nodes are highlighted in blue.
			Take Fig.~\ref{fig:nsb1} for example: the heavy nodes are $a$, $b$, $d$, and $d$;
			node $a$ has a degree of 3 and has a weight of 6; 
			node $b$ has a degree and a weight both equal to 3. 
			Note that although both nodes $a$ and $b$ are heavy nodes in time-slot 1, 
			the weight of $a$ equals twice the node degree because it was not scheduled in time-slot 1 (i.e., $U_a(2)=0$).
			The thick red lines denote the links activated in each time-slot.
			}
			}
        \label{fig:nsb}
\end{figure*}

\begin{figure*}[t!]
        \centering
        \begin{subfigure}[b]{0.21\linewidth}
                \includegraphics[width=\textwidth]{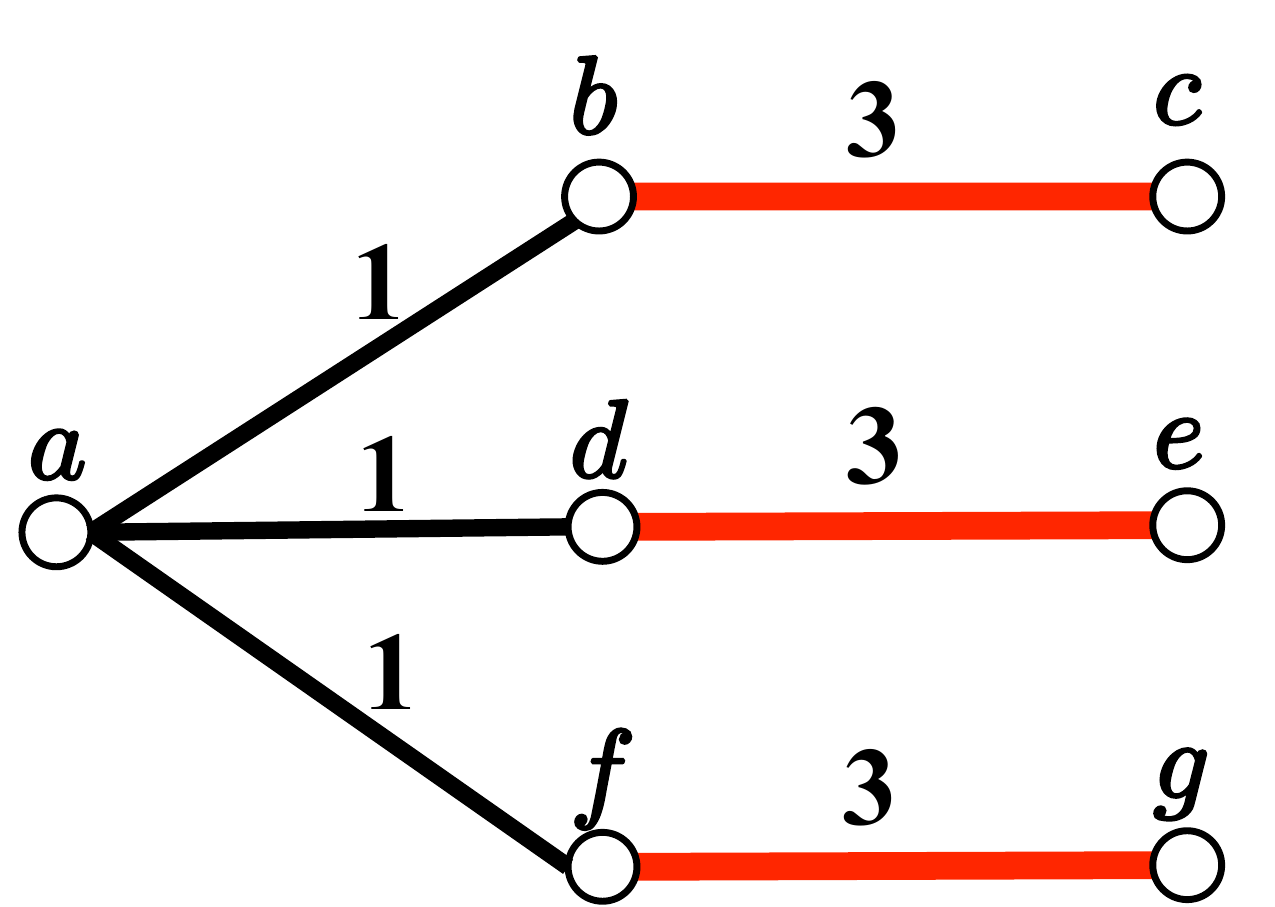}
                \caption{Time-slot 0}
                \label{fig:mwm0}
        \end{subfigure}%
        \quad
        \begin{subfigure}[b]{0.21\linewidth}
                \includegraphics[width=\textwidth]{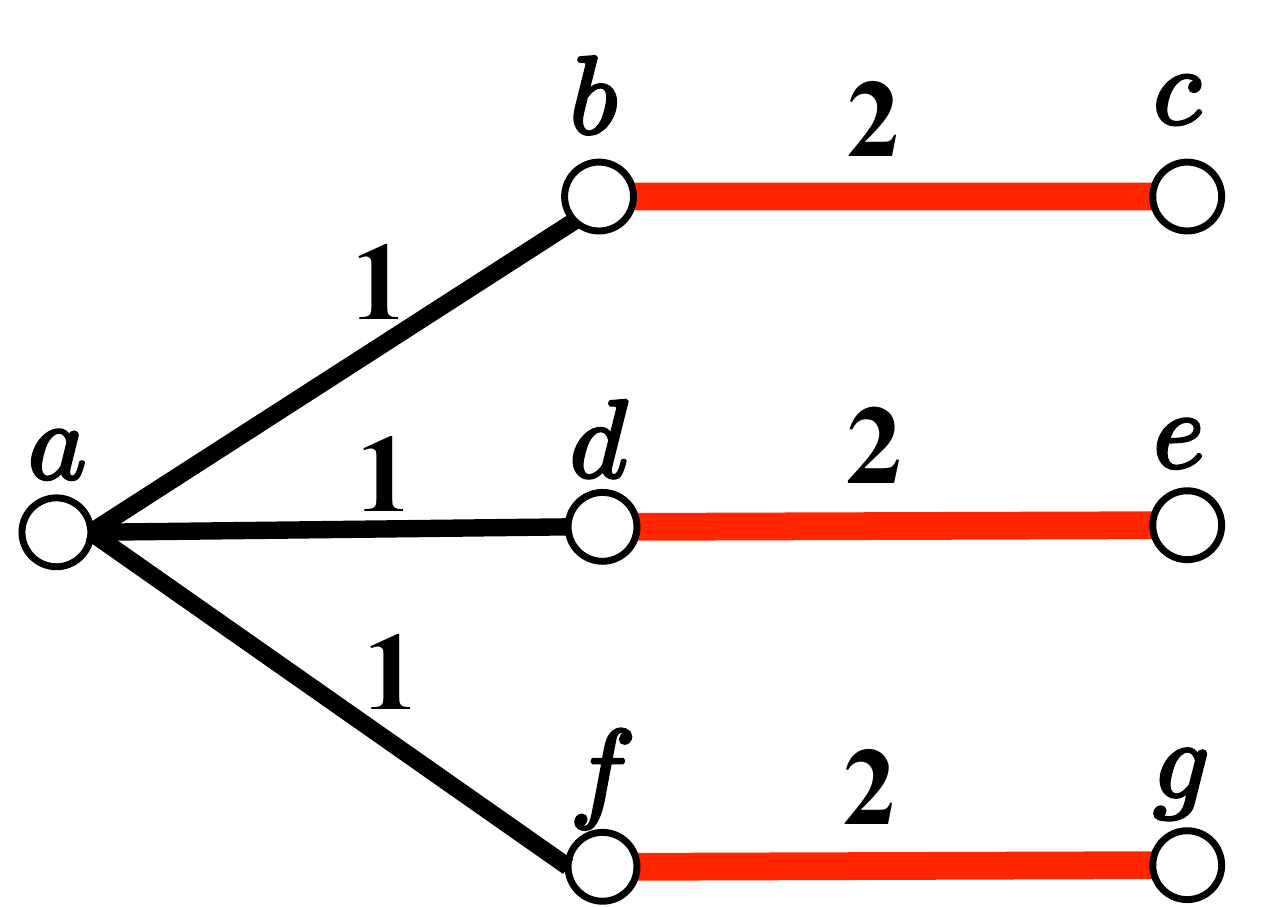}
                \caption{Time-slot 1}
                \label{fig:mwm1}
        \end{subfigure}
	\quad
	\begin{subfigure}[b]{0.21\linewidth}
                \includegraphics[width=\textwidth]{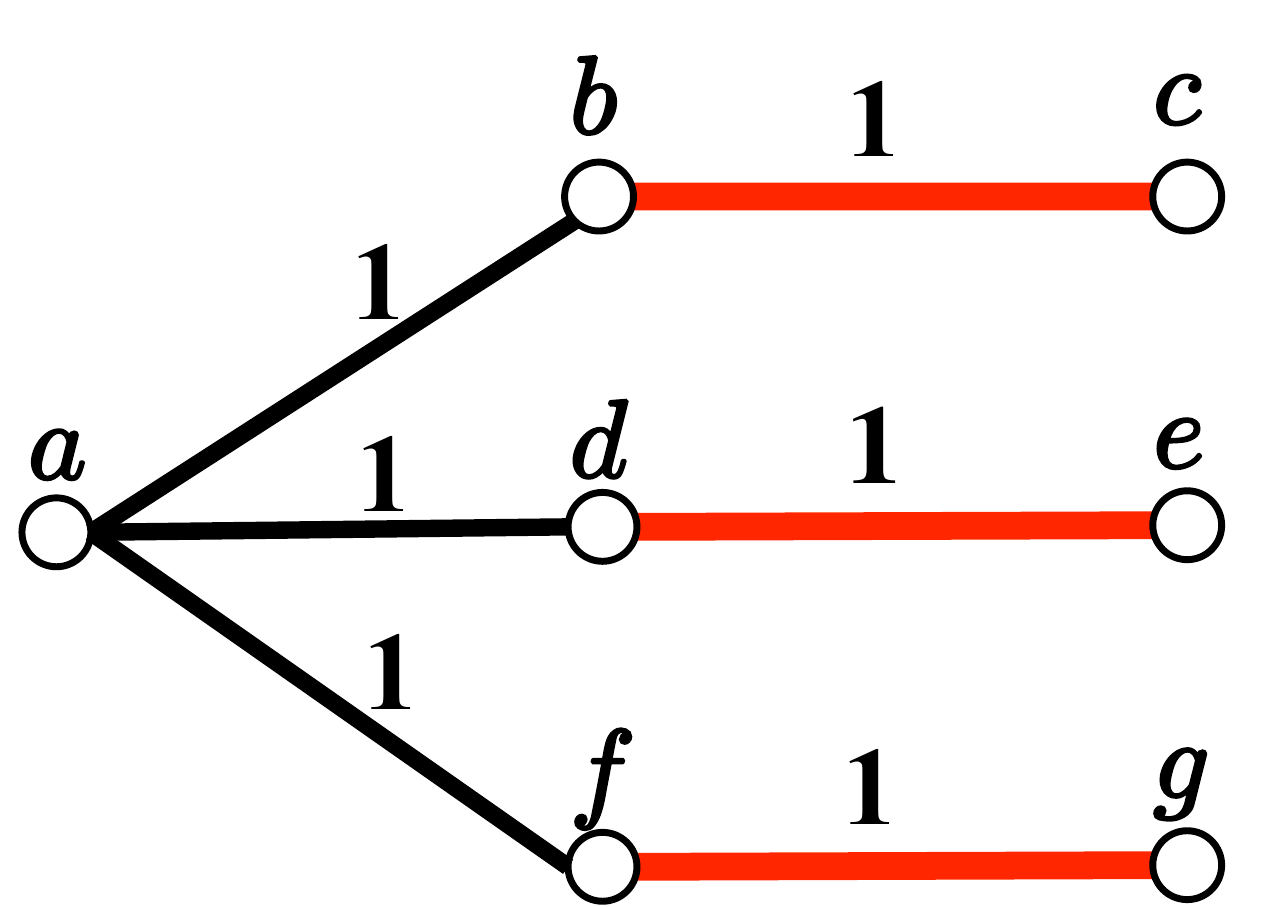}
                \caption{Time-slot 2}
                \label{fig:mwm2}
        \end{subfigure} 
        \quad
        \begin{subfigure}[b]{0.21\linewidth}
                \includegraphics[width=\textwidth]{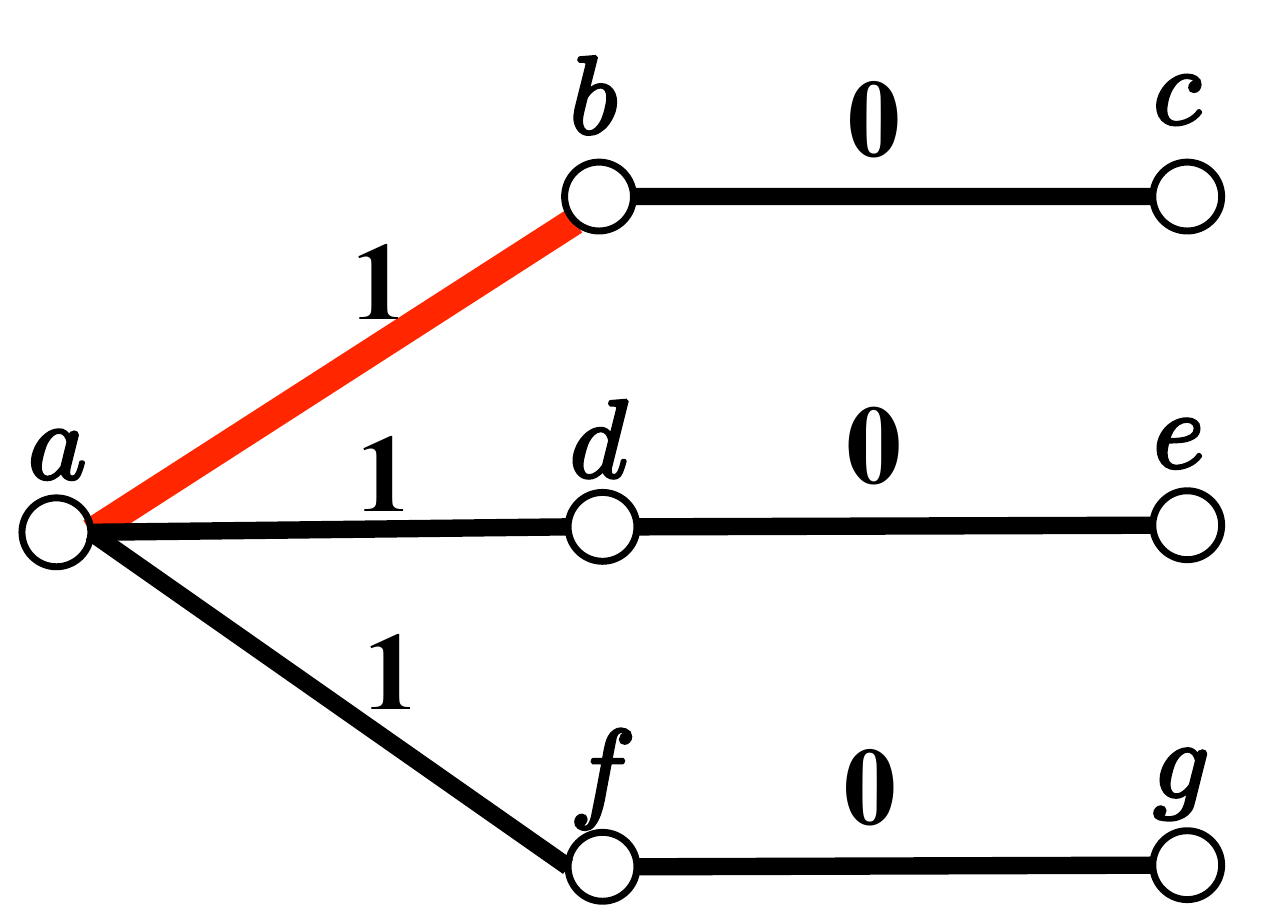}
                \caption{Time-slot 3}
                \label{fig:mwm3}
        \end{subfigure}%
        \caption{\HIGH{An illustration of the operations of MWM in four time-slots. The labels are similar to that in Fig.~\ref{fig:nsb}. 
       			The node degree and the node weight are not labeled because they are irrelevant under MWM.
			}
			}
        \label{fig:mwm}
\end{figure*}

Then, in time-slot $k$, we assign a weight to node $i$ as
\begin{equation}
\label{eq:weight}
w_i(k) \triangleq
\begin{cases}
Q_i(k)(2-U_i(k)) &\text{if $i \in \mH(k)$}; \\
Q_i(k) &\text{otherwise}.
\end{cases} 
\end{equation}
The NSB algorithm finds an MVM~\cite{spencer84} based on the assigned node weight $w_i(k)$'s in each time-slot. 
Note that links with a zero queue length will not be considered when MVM is computed. According to 
Eq. (\ref{eq:weight}), the nodes are divided into two groups: the heavy nodes that were not scheduled 
in the previous time-slot(s) have a weight twice their workload, and all the other nodes have a weight 
equal to their workload. Within each group, a node with a larger workload has a larger weight. 
\high{To help illustrate the operations of the NSB algorithm, we provide its pseudo code in Algorithm~\ref{alg:nsb}.}

\HIGH{
In addition, to help the reader better understand the operations of the NSB algorithm, we also present a simple example
in Fig.~\ref{fig:nsb}, which demonstrates the system evolution within four time-slots. Note that a certain tie-breaking rule 
is applied in this example. Using a different tie-breaking rule, different schedules could be selected. For instance, 
in time-slot 0 we may activate link $(a,b)$ 
instead of link $(b,c)$ along with links $(d,e)$ and $(f,g)$. However, tie-breaking rules do not affect our analysis.
As can be seen from Fig.~\ref{fig:nsb}, NSB drains all initial packets by the end of time-slot 3.
For comparison, using the same example, in Fig.~\ref{fig:mwm} we also demonstrate the system evolution under the 
MWM algorithm. Similarly, we observe that after time-slot 3, MWM needs two more time-slots to completely drain all
initial packets. In the next subsection, we will use a similar example to show that a link-based algorithm like MWM needs 
about twice as much time as that of a node-based algorithm like NSB to evacuate all initial packets in the network.
}


\subsection{Evacuation Time Performance} \label{subsec:nsb_et}
In this subsection, we analyze the evacuation time performance of NSB in the settings without 
arrivals. The main result is presented in Theorem~\ref{thm:nsb_et}.

\begin{theorem}
\label{thm:nsb_et}
The NSB algorithm has an approximation ratio no greater than $3/2$ for the evacuation time performance. 
\end{theorem}

\begin{proof}
Recall that for a given network with initial packets waiting to be transmitted, $\ChromaticIndex$
denotes the minimum evacuation time, and $T^{\text{NSB}}$ denotes the evacuation time of NSB. 
We want to show $T^{\text{NSB}} \le 3/2 \cdot \ChromaticIndex$. 
Recall that $\Delta(0)$ denotes the maximum node queue length in time-slot $0$.
If $\Delta(0)=1$, this is trivial as $T^{\text{NSB}} = \ChromaticIndex = 1$.
Now, suppose $\Delta(0) \ge 2$. Then, the result follows immediately from 1) Proposition~\ref{pro:nsb}
(stated after this proof): 
under NSB, the maximum node queue length decreases by at least two within \high{each frame}, i.e., 
$T^{\text{NSB}} \le 3/2 \cdot \Delta(0)$, and 2) an obvious fact: it takes at least $\Delta(0)$ time-slots 
to drain all the packets over the links incident to a node with maximum queue length, i.e., $\Delta(0) \le \ChromaticIndex$. 
\end{proof}

Next, we state a key proposition (Proposition~\ref{pro:nsb}) used for proving Theorem~\ref{thm:nsb_et}.

\begin{proposition}
\label{pro:nsb}
\high{Consider any frame.}
Suppose the maximum node queue length \high{is no smaller than two at the beginning of a frame}. 
Under the NSB algorithm, the maximum node queue length decreases by at least two \high{by the 
end of the frame}.
\end{proposition}

\HIGH{
We provide the detailed proof of Proposition~\ref{pro:nsb} in Section~\ref{sec:pro:nsb}
and give a sketch of the proof below.
Note that in any time-slot, the network together with the present packets can be 
represented as a loopless multigraph, where each multi-edge corresponds to a packet 
waiting to be transmitted over the link connecting the end nodes of the multi-edge.
We use $\Graph(k)$ to denote the multigraph at the beginning of time-slot $k$, 
and use $M(k)$ to denote the matching found by the NSB algorithm in time-slot $k$. 
Hence, the degree of node $i$ in $\Graph(k)$ is equivalent to the node queue length 
$Q_i(k)$, and the maximum node degree of $\Graph(k)$ is equal to $\Delta(k)$.
Now, consider any frame $k^{\prime}$ consisting of three consecutive time-slots 
$\{p, p+1, p+2\}$, where $p=3k^{\prime}$. 
Suppose that the maximum node queue length is no smaller than two at the beginning of 
frame $k^{\prime}$, i.e., $\Delta(p) \ge 2$ at the beginning of time-slot $p$.
Then, we want to show that under the NSB algorithm, the maximum degree will be 
at most $\Delta(p)-2$ at the end of time-slot $p+2$.
We proceed the proof in two steps: 
1) we first show that the maximum degree will decrease by at least one in the first 
two time-slots $p$ and $p+1$ (i.e., the maximum degree will be at most $\Delta(p)-1$
at the end of time-slot $p+1$),
and then, 2) show that if the maximum degree decreases by exactly one in the first 
two time-slots (i.e., the maximum degree is $\Delta(p)-1$ at the end of time-slot $p+1$), 
then the maximum degree must decrease by one in time-slot $p+2$, and becomes 
$\Delta(p)-2$ at the end of time-slot $p+2$. 
}

\high{\emph{Remark:} The key intuition that the NSB algorithm can provide provable evacuation time performance 
is that all the critical nodes are ensured to be scheduled at least twice within each frame (Proposition~\ref{pro:nsb}). 
This comes from the following properties of NSB: 
1) it results in the desired priority or ranking of the nodes by assigning the node weights according to Eq.~(\ref{eq:weight}); 
2) it finds an MVM based on the assigned node weights, and if a critical node was not scheduled in the first time-slot of the frame,
then in the second time-slot of the frame, MVM guarantees to match all such critical nodes (Lemma~\ref{lem:path}); 
3) similarly, if a critical node was not scheduled in both of the first two time-slots of the frame, 
then MVM guarantees to match all such critical nodes in the third time-slot of the frame.} 

\HIGH{In addition, we use an example to illustrate why a link-based algorithm like MWM could result in a bad evacuation time performance.
Consider the network topology presented in Fig.~\ref{fig:graphN}. Since MWM aims to maximize the total weight summed over all scheduled
links in each time-slot, it will choose the matching consisting of all the edges with $N$ packets. This pattern repeats until all links have one packet
(after $N-1$ time-slots). Then, it takes additional $N$ or $N+1$ time-slots to drain all the remaining packets, depending on the tie-breaking rule.
This results in inefficient schedules that consist of one link only about half the time, and thus, requires a total of $2N-1$ or $2N$ time-slots. 
Another link-based algorithm GMM performs similarly.
On the other hand, as we will show in Section~\ref{subsec:bipartite}, NSB is evacuation-time-optimal 
in this example and needs only $\Delta = N+1$ time-slots since the graph is bipartite. 
The system evolution under NSB and MWM for a special case of $N=3$ can be found in Figs.~\ref{fig:nsb} and \ref{fig:mwm}, respectively.
}

\begin{figure}[t!]
  \centering
    \includegraphics[width=0.25\textwidth]{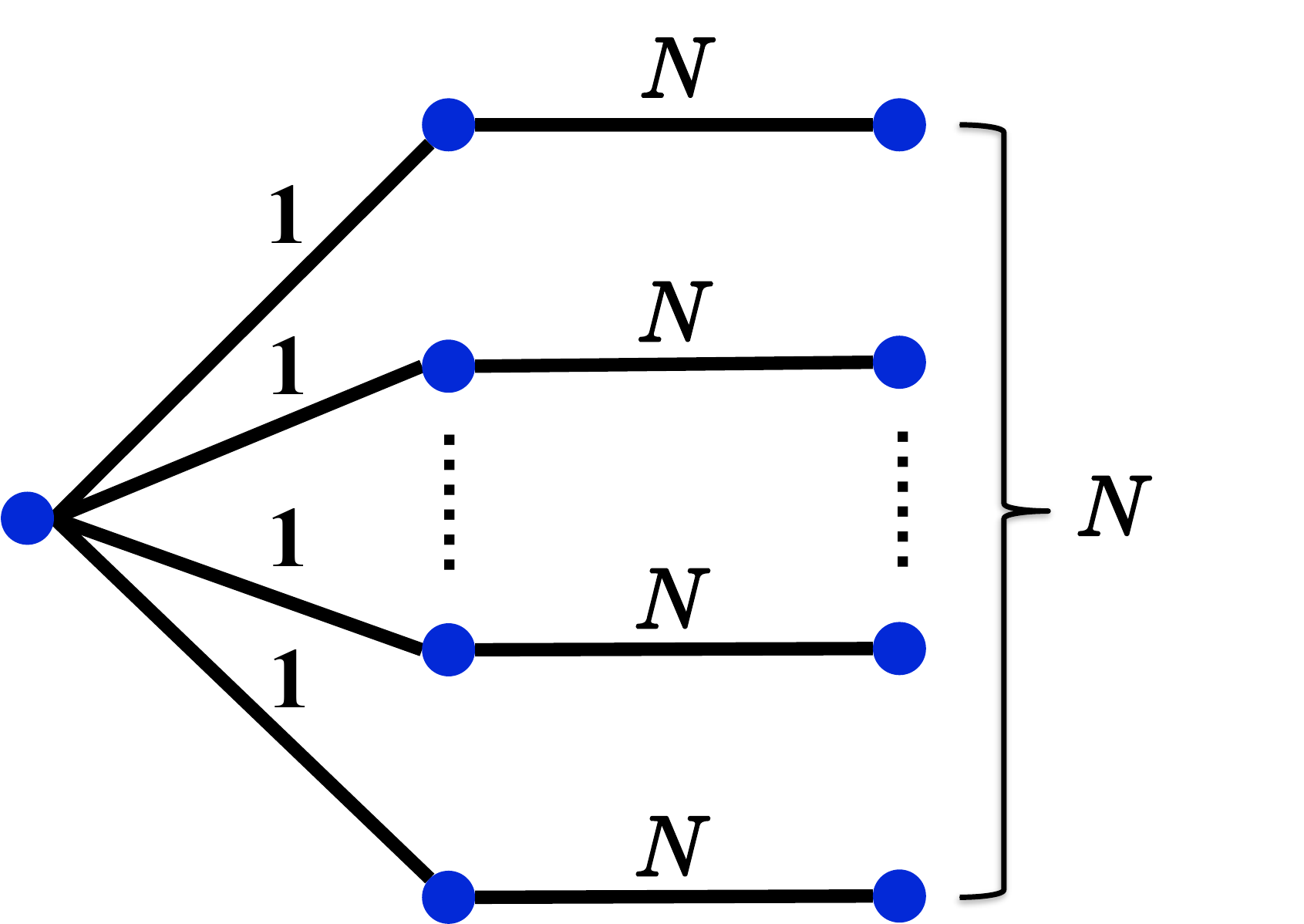}
  \caption{\high{A network with $2N+1$ nodes, where $N$ is a positive integer. 
  			The number above each link denotes the number of initial packets waiting to be transmitted over that link.
			}}
  \label{fig:graphN}
\end{figure}

\subsection{Throughput Performance} \label{subsec:nsb_throughput}
Next, we analyze the throughput performance of NSB in the settings with arrivals. The main 
result is presented in Theorem~\ref{thm:nsb_throughput}.

\begin{theorem}
\label{thm:nsb_throughput}
The NSB algorithm has an efficiency ratio no smaller than $2/3$ for the throughput performance. 
\end{theorem}

\high{We will employ fluid limit techniques \cite{dai95,dai00,andrews04} to prove Theorem~\ref{thm:nsb_throughput}. 
Fluid limit techniques are useful for two main reasons: 
(i) it removes irrelevant randomness in the original stochastic system such that the considered
system becomes deterministic and thus, the analysis can be simplified;
(ii) the algorithm exhibits some special properties in the fluid limit (e.g., Lemma~\ref{lem:2in3}), 
which do not exist in the original stochastic system.

Before proving Theorem~\ref{thm:nsb_throughput}, we construct the fluid model and 
state some definitions and a lemma that will be used in the proof. 
First, we extend the process $Y = A, Q, D, H$ to continuous time $t \ge 0$ by setting $Y(t) = Y(\lfloor t \rfloor)$.
Hence, $A(t)$, $Q(t)$, $D(t)$, and $H(t)$ are right continuous with left limits.
Then, using the techniques of \cite{dai00}}, we can show that for almost 
all sample paths and for all positive sequences $\xr \rightarrow \infty$, there exists a 
subsequence $\xrj$ with $\xrj \rightarrow \infty$ as $j \rightarrow \infty$ such that the 
following convergence holds uniformly over compact (\emph{u.o.c.}) intervals of time $t$:
\begin{eqnarray}
&&	\frac {A_i(\xrj t)} {\xrj} \rightarrow \lambda_i t,~\text{for all}~i \in \Vertex, \label{eq:fluid_a} \\ 
&&	\frac {Q_i(\xrj t)} {\xrj} \rightarrow q_i(t),~\text{for all}~i \in \Vertex, \\
&&	\frac {D_i(\xrj t)} {\xrj} \rightarrow d_i(t),~\text{for all}~i \in \Vertex, \\
&&	\frac {H_M(\xrj t)} {\xrj} \rightarrow h_M(t),~\text{for all}~M \in \Matching. \label{eq:fluid_h}
\end{eqnarray}
\HIGH{Since the proof of the above convergence is standard, we provide the proof in the appendix for completeness.}

Next, we present the fluid model equations as follows:
\begin{eqnarray}
&& q_i(t) = q_i(0) + \lambda_i t - d_i(t),~\text{for all}~i \in \Vertex, \label{eq:qqld} \\
&& \diff d_i(t) = \sum_{M \in \Matching} \sum_{l \in L(i)} M_l \cdot \diff h_M(t),~\text{for all}~i \in \Vertex, \label{eq:dd}\\
&& \sum_{M \in \Matching} h_M(t) = t.
\end{eqnarray}
Any such limit $(q,d,h)$ is called a \emph{fluid limit}. Note that $q_i(\cdot), d_i(\cdot)$ and 
$h_i(\cdot)$ are absolutely continuous functions and are differentiable at almost all times 
$t \ge 0$ (called \emph{regular} times). Taking the derivative of both sides of (\ref{eq:qqld})
and substituting (\ref{eq:dd}) into it, we obtain
\begin{equation}
\label{eq:diff}
\begin{split}
\diff q_i(t) &= \lambda_i - \diff d_i(t) \\
&= \lambda_i - \sum_{M \in \Matching} \sum_{l \in L(i)} M_l \cdot \diff h_M(t).
\end{split}
\end{equation}


Borrowing the results of \cite{dai00}, we give the definition of \emph{weak stability} 
and state Lemma~\ref{lem:weak2rate}, which establishes the connection between rate stability 
of the original system and weak stability of the fluid model.

\begin{definition}
\label{def:weak}
The fluid model of a network is weakly stable if for every fluid model solution $(q,d,h)$ 
with $q(0)=0$, one has $q(t)=0$ for all regular times $t \ge 0$.
\end{definition}

\begin{lemma}[\high{Theorem 3 of \cite{dai00}}]
\label{lem:weak2rate}
A network is rate stable if the associated fluid model is weakly stable.
\end{lemma}

We are now ready to prove Theorem~\ref{thm:nsb_throughput}.

\begin{proof}[Proof of Theorem~\ref{thm:nsb_throughput}]
We want to show that given any arrival rate vector $\lambda$ strictly inside $2/3 \cdot \Lambda^*$, 
the system is rate stable under the NSB algorithm. Note that $\lambda$ is also strictly inside 
$2/3 \cdot \Psi$ (i.e., $\lambda_i < 2/3$ for all $i \in \Vertex$) since $\Lambda^* \subseteq \Psi$. 
We define $\epsilon \triangleq \min_{i \in \Vertex} (2/3-\lambda_i)$. Clearly, we must have $\epsilon > 0$. 

To show rate stability of the original system, it suffices to show weak stability of the fluid model 
due to Lemma~\ref{lem:weak2rate}. We start by defining the following Lyapunov function:
\begin{equation}
V(q(t)) = \max_{i \in \Vertex} q_i(t).
\end{equation}
\HIGH{For any regular time $t \ge 0$, we define the drift of $V(q(t))$ as its derivative, denoted by $\diff V(q(t))$.}
Since $V(q(t))$ is a non-negative function, given $q(0)=0$, in order to show $V(q(t))=0$ and 
thus $q(t)=0$ for all regular times $t \ge 0$, it suffices to show that if $V(q(t))>0$ for $t>0$, 
then $V(q(t))$ has a negative drift.
\high{This is due to a simple result in Lemma 1 of \cite{dai00}.} 
Therefore, we want to show that for all regular times $t>0$, if $V(q(t)) > 0$, then $\diff V(q(t)) \le - \epsilon$.

We first fix time $t$ and let $q_{\max}(t)=V(q(t))=\max_{i \in \Vertex} q_i(t)$. Define the set of 
critical nodes in the fluid limits at time $t$ as
\begin{equation}
\label{eq:critical}
\mC \triangleq \{i \in \Vertex ~|~ q_i(t) = q_{\max}(t) \}.
\end{equation}
Also, let $\hat{q}_{\max}(t)$ be the largest queue length in the fluid limits among the remaining 
nodes, i.e., $\hat{q}_{\max}(t) = \max_{i \in \Vertex \backslash \mC} q_i(t)$. Since the number of 
nodes is finite, we have $\hat{q}_{\max}(t) < q_{\max}(t)$. Choose $\beta$ small enough such that 
$\hat{q}_{\max}(t) < q_{\max}(t) - 3 \beta$ and $\beta < q_{\max}(t)/(2n-1)$. Our choice of $\beta$ 
implies the following:
\begin{equation}
\label{eq:heavy}
q_{\max}(t)-\beta > \frac{n-1}{n} (q_{\max}(t) + \beta).
\end{equation}

Recall that $q(t)$ is absolutely continuous. Hence, there exists a small $\delta$ such that the 
queue lengths in the fluid limits satisfy the following condition for all times $\tau \in (t,t+\delta)$ 
\HIGH{and for all nodes $i \in \Vertex$:
\begin{equation}
\label{eq:qtau}
q_i(\tau) \in \left( q_i(t)-\beta/2, q_i(t)+\beta/2 \right).
\end{equation}
This further implies that the following conditions hold:

(C1) $q_i(\tau) \in (q_{\max}(t)-\beta/2, q_{\max}(t)+\beta/2)$ for all $i \in \mC$;

(C2) $q_i(\tau) < q_{\max}(t)-5\beta/2$ for all $i \notin \mC$,

\noindent where (C2) is from Eq.~\eqref{eq:qtau} and $q_i(t) \le \hat{q}_{\max}(t)<q_{\max}(t) - 3 \beta$ for all $i \notin \mC$.
}

Let $\xrj$ be a positive subsequence for which the convergence to the fluid limit holds.
Consider a large enough $j$ such that $\left \lvert Q_i(\xrj \tau)/\xrj - q_i(\tau) \right \rvert < \beta/2$ 
for all $\tau \in (t,t+\delta)$.
Considering the interval $(t,t+\delta)$ around time $t$, we define a set of consecutive
time-slots in the original system as
$T \triangleq \{\lceil \xrj t \rceil, \lceil \xrj t \rceil + 1, \dots, \lfloor \xrj (t+\delta) \rfloor\}$,
which corresponds to the scaled time interval $(t,t+\delta)$ in the fluid limits.

Lemma~\ref{lem:2in3} states that NSB, all the critical nodes at scaled time $t$ in the fluid 
limits will be scheduled at least twice within \high{each frame of interval $T$}. 

\begin{lemma}
\label{lem:2in3}
Under the NSB algorithm, all the nodes in $\mC$ will be scheduled at least twice within 
\high{each frame of interval $T$}.
\end{lemma}

We provide the proof of Lemma~\ref{lem:2in3} in Section~\ref{sec:lem:2in3}.
For now, we assume that Lemma~\ref{lem:2in3} holds. \high{Note that interval $T$ 
contains at least $(\lfloor \xrj (t+\delta) \rfloor - \lceil \xrj t \rceil - 3)/3$ complete frames.}
Then, from Lemma~\ref{lem:2in3}, we have that for all $i \in \mC$,
\begin{equation}
\label{eq:beta}
\begin{split}
& \sum_{M \in \Matching} \sum_{l \in L(i)} M_l \cdot (H_M(\xrj (t+\delta))-H_M(\xrj t)) \\
& \ge 2/3 \cdot (\lfloor \xrj (t+\delta) \rfloor - \lceil \xrj t \rceil - 3),
\end{split}
\end{equation}
and therefore, we have
\begin{equation}
\label{eq:dh}
\begin{split}
&\sum_{M \in \Matching} \sum_{l \in L(i)} M_l \cdot \diff h_M(t) \\
&= \lim_{\delta \rightarrow 0} \sum_{M \in \Matching} \sum_{l \in L(i)} M_l \cdot \frac{h_M(t+\delta)-h_M(t)}{\delta} \\
&\stackrel{(a)} = \lim_{\delta \rightarrow 0} \lim_{j \rightarrow \infty} \sum_{M \in \Matching} \sum_{l \in L(i)} \frac{M_l \cdot (H_M(\xrj (t+\delta))-H_M(\xrj t))}{\xrj \delta} \\
&\stackrel{(b)} \ge \lim_{\delta \rightarrow 0} \lim_{j \rightarrow \infty} \frac{2/3 \cdot (\lfloor \xrj (t+\delta) \rfloor - \lceil \xrj t \rceil - 3)}{\xrj \delta} \\
&= 2/3,
\end{split}
\end{equation}
where $(a)$ is from the convergence in Eq.~\eqref{eq:fluid_h} and $(b)$ is from Eq.~\eqref{eq:beta}. 
Then, it follows from Eq. (\ref{eq:diff}) that for all $i \in \mC$, we have $\diff q_i(t) \le \lambda_i - 2/3 \le -\epsilon$.

Also, from conditions (C1) and (C2), every node $i \notin \mC$ has a queue length strictly smaller 
than that of a critical node in $\mC$ for the entire duration $(t, t+\delta)$. Thus, we have
$\diff V(q(t)) \le - \epsilon$, which implies that the fluid model is weakly stable. Then, we complete 
the proof by applying Lemma~\ref{lem:weak2rate}.
\end{proof}

\emph{Remark:} The key intuition that the NSB algorithm can provide provable throughput performance 
is that all the critical nodes in the fluid limit are ensured to be scheduled at least twice within each frame of 
interval $T$ (Lemma~\ref{lem:2in3}). Similar to the provable evacuation time performance, this comes from 
the desired weight assignment (Eq.~(\ref{eq:weight})) and an important property of MVM (Lemma~\ref{lem:path}).

As we mentioned earlier, the proof of Lemma~\ref{lem:2in3} is provided in Section~\ref{sec:lem:2in3}. 
However, we want to emphasize that the proof relies on a novel application of graph-factor theory, 
which is stated in Lemma~\ref{lem:existence}.
Lemma~\ref{lem:existence} is of critical importance in proving the guaranteed
throughput performance of the NSB algorithm in general graphs. Moreover, it will play a key 
role in establishing both throughput optimality and evacuation time optimality for NSB in bipartite 
graphs (see Section~\ref{subsec:bipartite}).

Next, we give some additional notations that are needed to state Lemma~\ref{lem:existence}.
By slightly abusing the notations, we also use $\Graph=(\Vertex, \Edge)$ to denote a multigraph,
which is possibly not loopless. 
Let $d_G(v)$ be the degree of node $v$ in $G$, where a loop associated with node $v$ counts 2 
towards the degree of $v$. 
Also, let $\Graph_Z$ denote the subgraph of $\Graph$ induced by a subset of nodes $Z \subseteq \Vertex$. 

\begin{lemma}
\label{lem:existence}
Let $\Graph=(\Vertex, \Edge)$ be a multigraph with maximum degree $\Delta$.
Consider a subset of nodes $Z \subseteq \Vertex$.
Suppose that the following conditions are satisfied: 
(i) all the nodes of $Z$ are heavy nodes, i.e., $Z \subseteq \{v \in \Vertex ~|~ d_{\Graph}(v) \ge (n-1)/n \cdot \Delta\}$, 
and (ii) $\Graph_Z$ is bipartite. Then, there exists a matching of $\Graph$ that matches every node of $Z$.
\end{lemma}

\begin{proof}
We first introduce some additional notations.
Let $g=[g_v: v \in \Vertex]$ and $f=[f_v: v \in \Vertex]$ be vectors of positive integers satisfying
\begin{equation}
0 \le g_v \le f_v \le d_G(v),~\text{for all}~ v \in \Vertex.
\end{equation}
A $(g,f)$-factor is a subgraph $F$ of $\Graph$ with
\begin{equation}
g_v \le d_F(v) \le f_v ,~\text{for all}~ v \in \Vertex.
\end{equation}
Note that if vectors $g$, $f$ satisfy $g_v,f_v \in \{0,1\}$ for all $v \in \Vertex$,
then the edges of a $(g,f)$-factor form a matching of $\Graph$. 
Let $\Graph_{g=f}$ denote the subgraph of $\Graph$ induced by nodes $v$ for which $g_v=f_v$,
and let $[x]^+ \triangleq  \max \{x,0\}$. We restate a result of \cite{anstee90} in Lemma~\ref{lem:gf}, 
which will be used in the proof.

\begin{lemma}[Theorem~1.3 of \cite{anstee90}, Property I]
\label{lem:gf}
Let multigraph $\Graph$ and vectors $g$, $f$ be given. Suppose that $\Graph_{g=f}$ is bipartite. 
Then, $\Graph$ has a $(g,f)$-factor if and only if for all node subsets $S \subseteq \Vertex$,
\begin{equation}
\label{eq:gf}
\sum_{v \in S} f_v \ge \sum_{v \notin S} [g_v - d_{\Graph-S}(v)]^+.
\end{equation}
\end{lemma}

\begin{figure}[t!]
  \centering
    \includegraphics[width=0.3\textwidth]{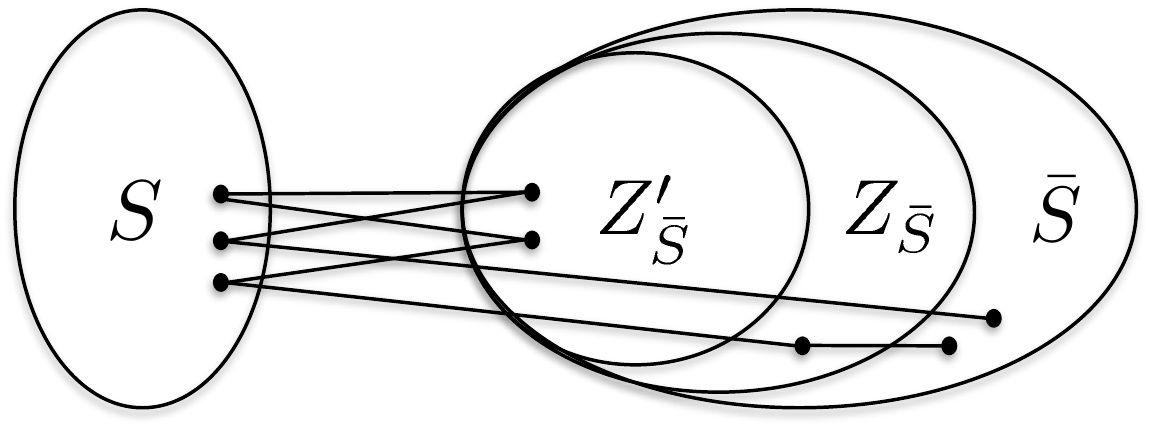}
  \caption{An illustration for the relationship of the sets in the proof of Lemma~\ref{lem:existence}.}
  \label{fig:set}
\end{figure}

%

We are now ready to prove Lemma~\ref{lem:existence}. 
Suppose that (i) all the nodes of $Z$ are heavy nodes and (ii) $\Graph_{Z}$ is bipartite. 
We construct vectors $g$, $f$ by setting $g_v=f_v=1$ for node $v \in Z$ and $g_v=0$, $f_v=1$ otherwise. 
With our constructed vectors $g$ and $f$, not only a $(g,f)$-factor forms a matching of graph $\Graph$, 
but this matching also matches every node of $Z$.
Therefore, it suffices to show that $\Graph$ has a $(g,f)$-factor.

Next, we apply Lemma~\ref{lem:gf} to show that $\Graph$ has a $(g,f)$-factor.
Note that $\Graph_{g=f} = \Graph_{Z}$ and $\Graph_{Z}$ is bipartite.
Then, it remains to show that Eq.~(\ref{eq:gf}) is satisfied for any subset of nodes $S \subseteq \Vertex$. 
Let $\bar{S} = \Vertex \backslash S$ be the complementary set of $S$. Let 
$Z_{\bar{S}} = Z \cap \bar{S}$, and let $Z_{\bar{S}}^{\prime} = \{ v \in Z_{\bar{S}} ~|~$ all 
the neighboring nodes of $v$ are in $S \}$. The relationship of these sets is illustrated in 
Fig.~\ref{fig:set}. Clearly, we have $\sum_{v \in S} f_v = \lvert S \rvert$ as $f_v=1$ for 
all $v \in \Vertex$. Also, a little thought gives 
$\sum_{v \notin S} [g_v - d_{G-S}(v)]^+ = \lvert Z_{\bar{S}}^{\prime} \rvert$. This is because 
any node $v \notin S$ must belong to one of the following three cases:
\begin{enumerate}
\item If $v \notin Z$, then $g_v=0$, and thus, $[g_v - d_{G-S}(v)]^+ = 0$; 

\item If $v \in Z_{\bar{S}} \backslash Z_{\bar{S}}^{\prime}$, then $g_v=1$ and $d_{G-S}(v) \ge 1$, 
which implies $[g_v - d_{G-S}(v)]^+ = 0$;

\item If $v \in Z_{\bar{S}}^{\prime}$, then $g_v=1$ and $d_{G-S}(v)=0$, which implies $[g_v - d_{G-S}(v)]^+ = 1$.
\end{enumerate}

Hence, in order to show Eq.~(\ref{eq:gf}), it remains to show 
$\lvert S \rvert \ge \lvert Z_{\bar{S}}^{\prime} \rvert$.
We prove it by contradiction. Suppose $\lvert S \rvert < \lvert Z_{\bar{S}}^{\prime} \rvert$.
Since $S$ and $Z_{\bar{S}}^{\prime}$ are disjoint, we have 
$\lvert S \rvert  + \lvert Z_{\bar{S}}^{\prime} \rvert \le n$, and thus, $\lvert S \rvert < n$. 
We let $d_{\Graph}(Z) = \sum_{i \in Z} d_{\Graph}(i)$ denote the total degree of a subset 
of nodes $Z \subseteq \Vertex$ in $\Graph$. Then, we state three obvious facts:

(F1) $d_{\Graph}(Z_{\bar{S}}^{\prime}) \ge (n-1)/n \cdot \Delta \lvert  Z_{\bar{S}}^{\prime} \rvert$;

(F2) $d_{\Graph}(S) \le \Delta \lvert S \rvert$;

(F3) $d_{\Graph}(Z_{\bar{S}}^{\prime}) \le d_{\Graph}(S)$.

\noindent Note that (F1) is from the fact that every node of $Z$ has a degree no smaller than 
$(n-1)/n \cdot \Delta$, (F2) is trivial, and (F3) is from the definition of $Z_{\bar{S}}^{\prime}$ 
that all of its neighboring nodes belong to $S$. Then, by combining the above facts, we obtain
\begin{equation}
\label{eq:sh}
\frac{n-1}{n} \le \frac{\lvert S \rvert}{\lvert  Z_{\bar{S}}^{\prime} \rvert}.
\end{equation}
This further implies $\lvert  Z_{\bar{S}}^{\prime} \rvert - \lvert S \rvert \le \lvert S \rvert/(n-1) \le 1$,
as $\lvert S \rvert < n$. Hence, we must have $\lvert  Z_{\bar{S}}^{\prime} \rvert = \lvert S \rvert + 1$,
because $\lvert S \rvert < \lvert Z_{\bar{S}}^{\prime} \rvert$. Substituting this back into Eq.~(\ref{eq:sh})
gives $(n-1)/n \le \lvert S \rvert / (\lvert S \rvert + 1)$. This implies $\lvert S \rvert \ge n-1$,
and thus $\lvert S \rvert  + \lvert Z_{\bar{S}}^{\prime} \rvert \ge 2n-1$, which contradicts the fact that 
$\lvert S \rvert  + \lvert Z_{\bar{S}}^{\prime} \rvert \le n$. (We assume $n>1$ to avoid trivial discussions.)

Therefore, we have $\lvert S \rvert \ge \lvert Z_{\bar{S}}^{\prime} \rvert$, and thus, Eq.~(\ref{eq:gf})
is satisfied. Then, Lemma~\ref{lem:gf} implies that graph $\Graph$ has a $(g,f)$-factor, which forms 
a matching of graph $\Graph$ that matches every node of $Z$. This completes the proof.
\end{proof}

\subsection{Optimality in Bipartite Graphs} \label{subsec:bipartite}
\high{As we have explained in the introduction, the minimum evacuation time is lower 
bounded by the largest workload at the nodes and the odd-size cycles. Hence, both the nodes 
and the odd-size cycles could be bottlenecks. Ideally, it would be best to consider the workload 
of both the nodes and the odd-size cycles when making scheduling decisions. However, this may 
render the algorithm very complex because it is much more difficult to explicitly consider all the 
odd-size cycles in a graph. Hence, we have focused on designing the node-based algorithms 
(such as the NSB algorithm) that do not explicitly handle the odd-size cycles. 
Without considering all the bottlenecks, these algorithms may not be able to achieve the best 
achievable performance in general. However, the theoretical results of Theorems~\ref{thm:nsb_et} 
and \ref{thm:nsb_throughput} are quite remarkable in the sense that even without considering the 
odd-size cycles (which could also be bottlenecks), the NSB algorithm can guarantee an 
approximation ratio no greater than $3/2$ for the evacuation time and an efficiency ratio 
no smaller than $2/3$ for the throughput.
We believe that NSB will perform better if odd-size cycles do not form bottlenecks. 
This is also observed from our simulation results in Section~\ref{sec:sim}.}

In this subsection, we will show that NSB is both throughput-optimal and evacuation-time-optimal in bipartite 
graphs, \high{where there are no odd-size cycles}. This result is stated in Theorem~\ref{thm:nsb_opt},
whose proof needs to apply Lemmas~\ref{lem:existence} and \ref{lem:path} and follows 
a similar line of analysis to that of Theorems~\ref{thm:nsb_et} and \ref{thm:nsb_throughput} for general graphs.
\HIGH{The detailed proof is provided in the appendix for completeness.}

\begin{theorem}
\label{thm:nsb_opt}
The NSB algorithm is both throughput-optimal and evacuation-time-optimal in bipartite graphs. 
\end{theorem}

\emph{Remark:} The NSB algorithm has a complexity of $O(m \sqrt{n} \log n)$, as the complexity of finding 
an MVM is $O(m \sqrt{n} \log n)$ \cite{spencer84}. One important question is whether we can develop 
lower-complexity algorithms that provide the same performance guarantees. We answer this question in the next section.

\begin{figure*}[t!]
        \centering
        \begin{subfigure}[b]{0.3\linewidth}
                \includegraphics[width=\textwidth]{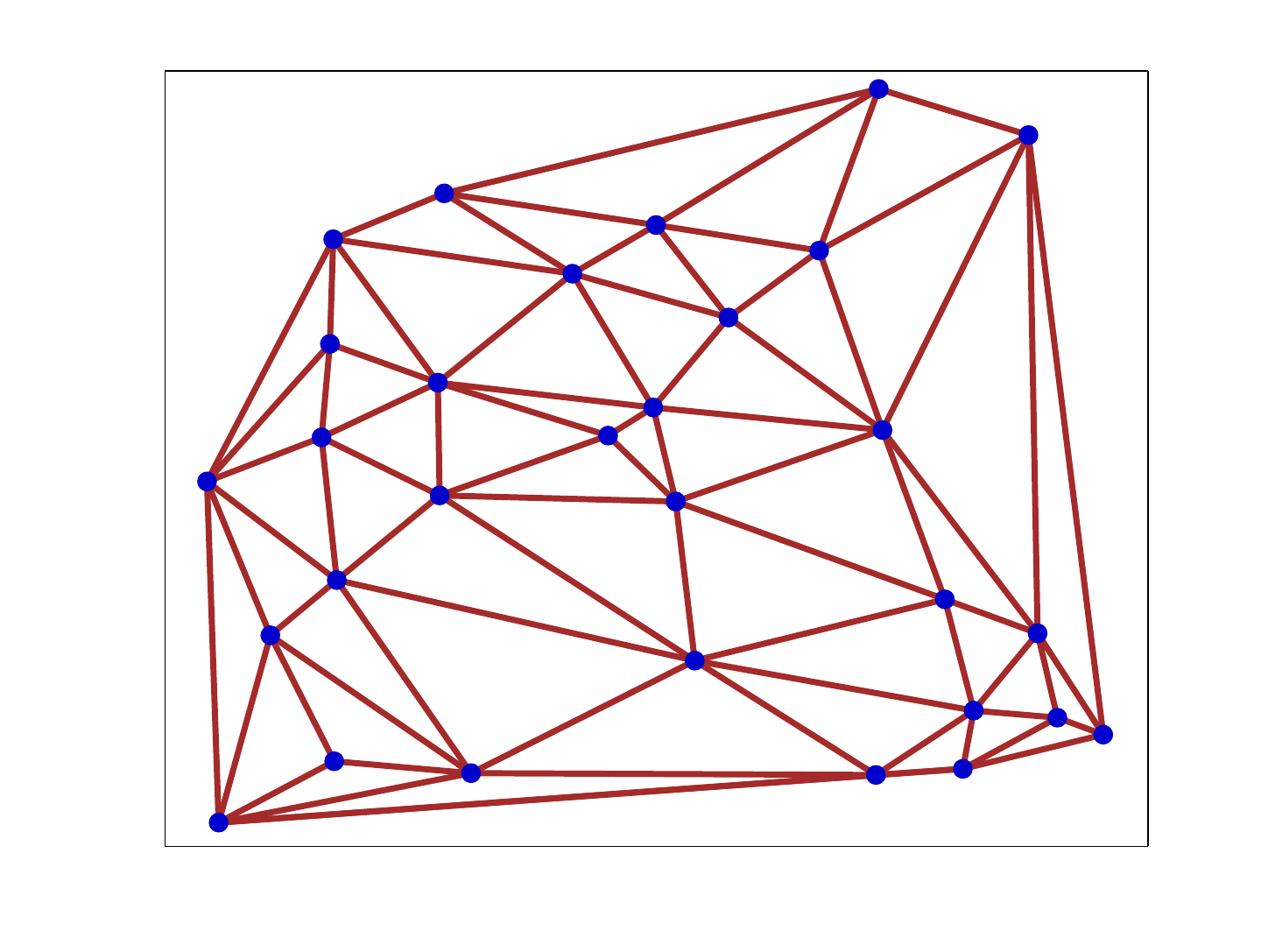}
                \caption{Triangular Mesh}
                \label{fig:triangle_top}
        \end{subfigure}%
        \qquad
        \begin{subfigure}[b]{0.3\linewidth}
                \includegraphics[width=\textwidth]{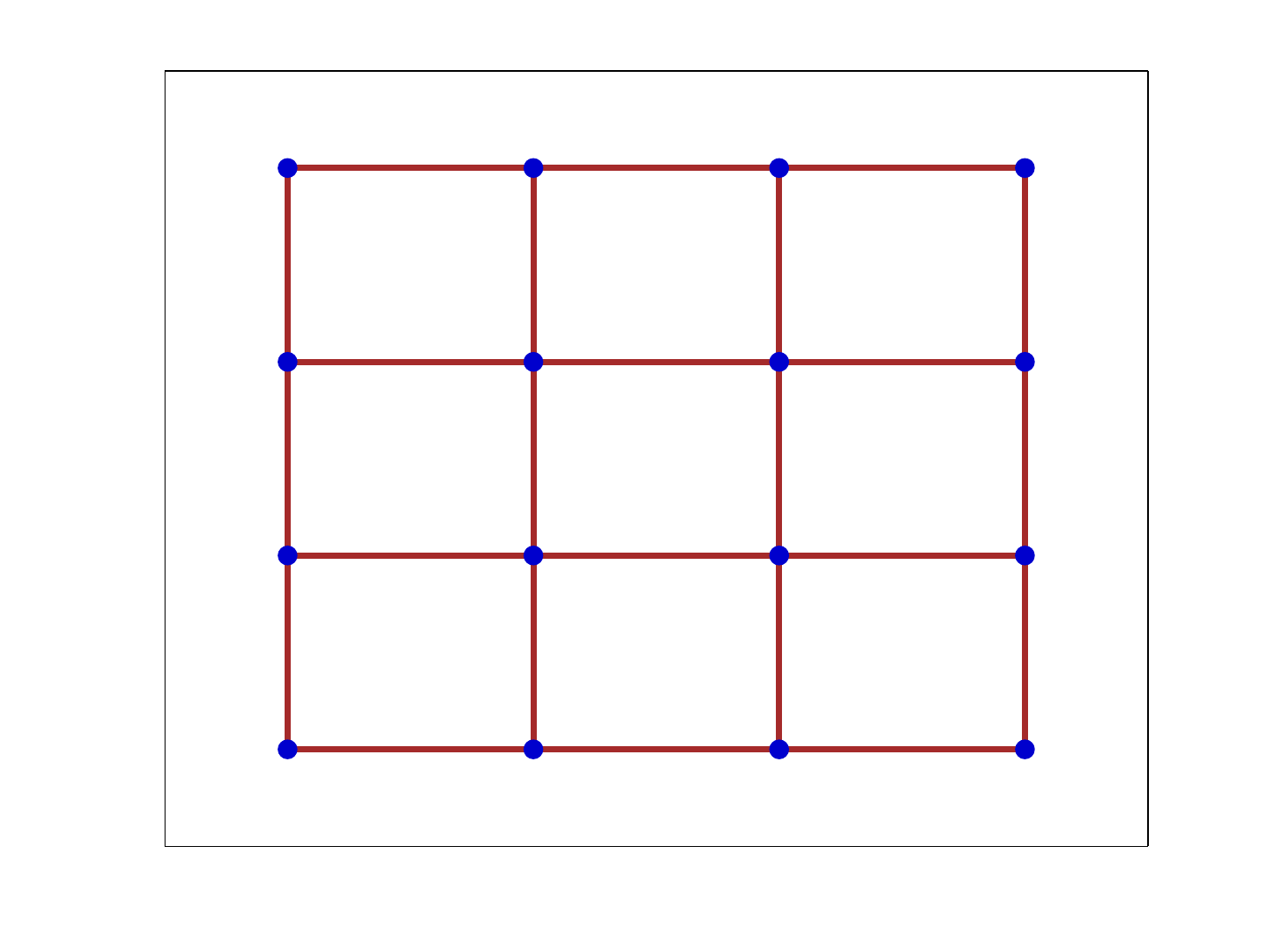}
                \caption{4$\times$4 grid}
                \label{fig:grid_top}
        \end{subfigure}
	\qquad
	\begin{subfigure}[b]{0.3\linewidth}
                \includegraphics[width=\textwidth]{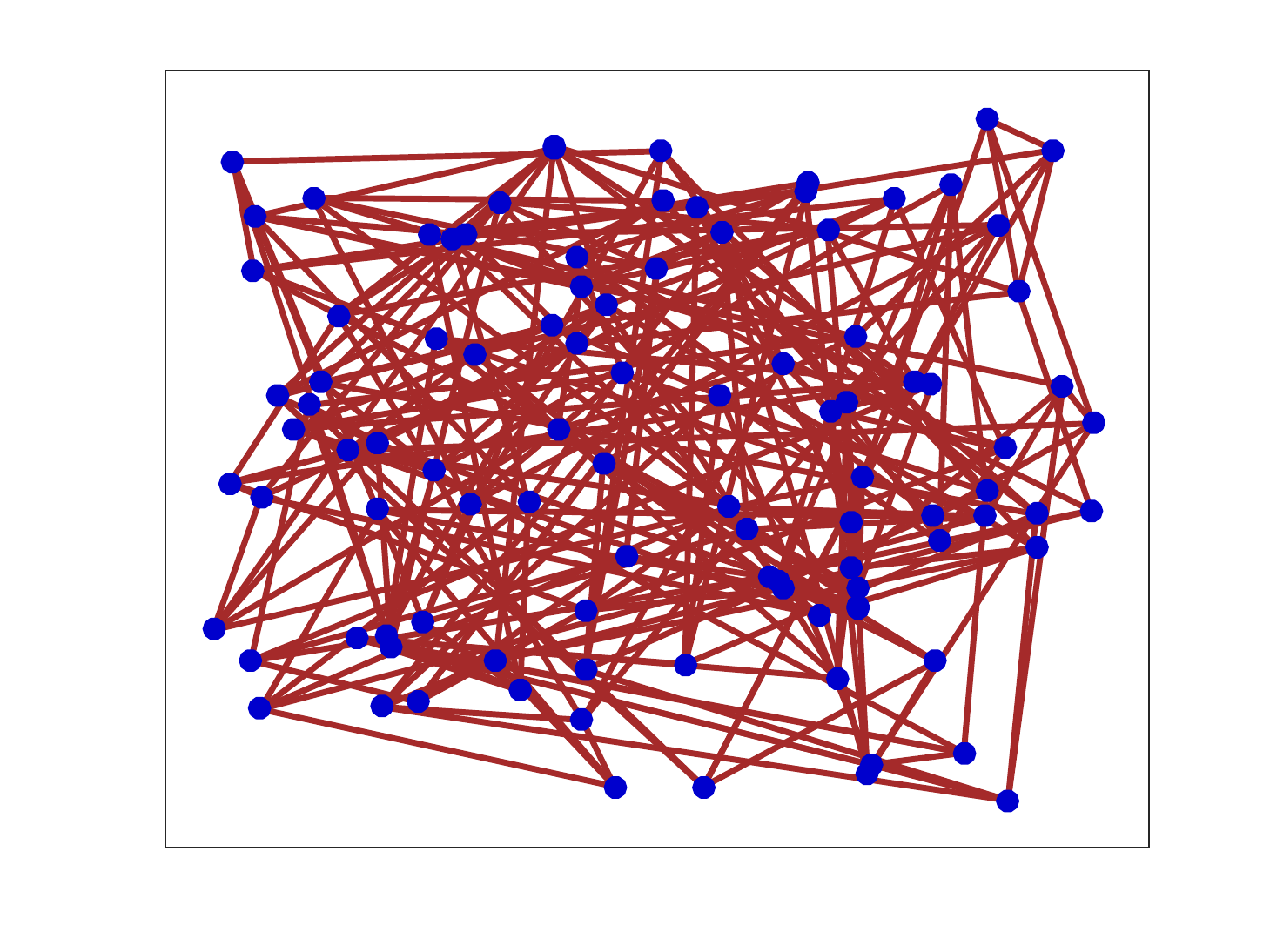}
                \caption{Random topology}
                \label{fig:random_top}
        \end{subfigure}  
        \caption{Simulation Topologies.}
        \label{fig:topologies}
\end{figure*}

\begin{figure*}[t!]
        \centering
        \begin{subfigure}[b]{0.3\linewidth}
                \includegraphics[width=\textwidth]{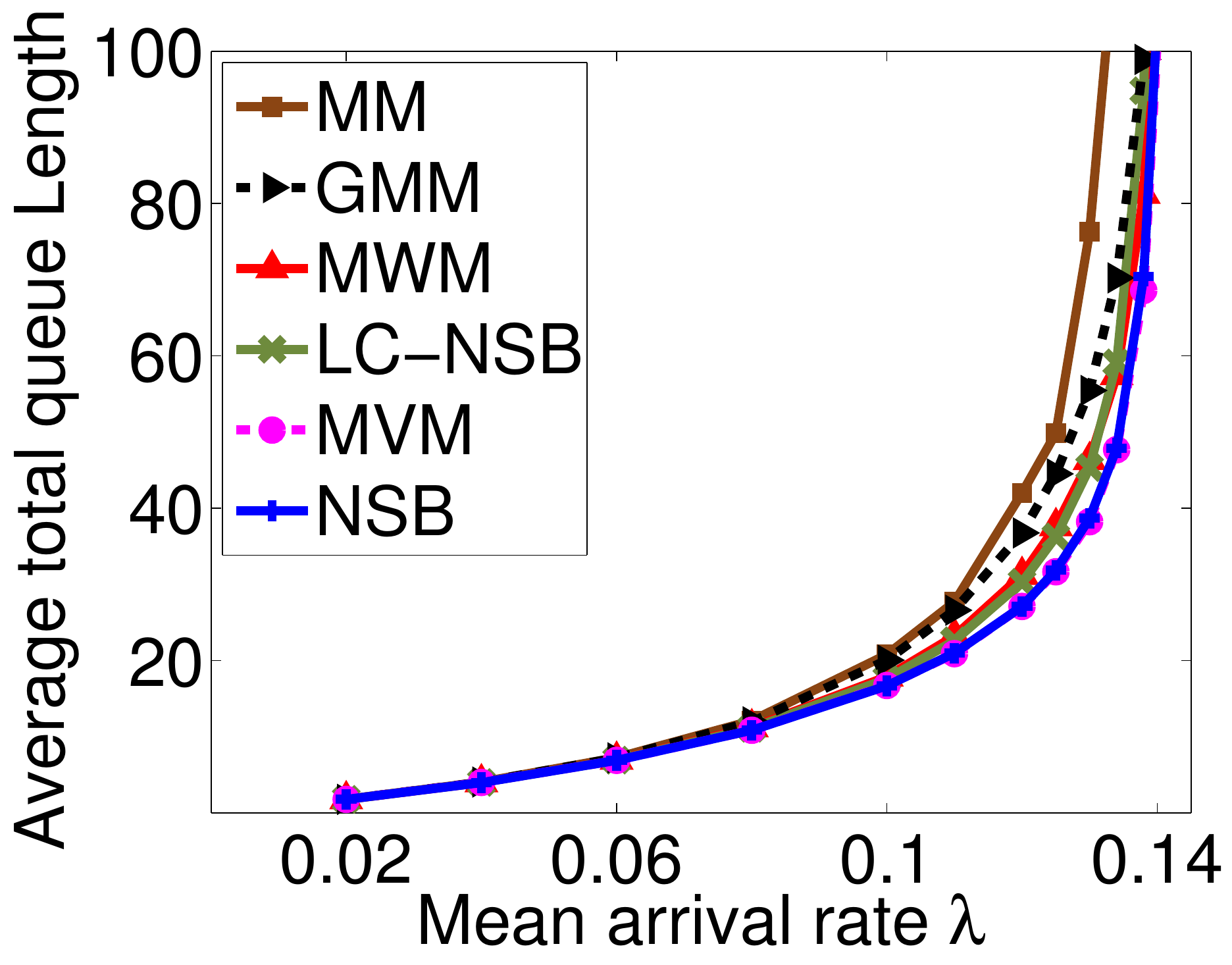}
                \caption{Triangular Mesh}
                \label{fig:triangle_result}
        \end{subfigure}%
       	\quad
        \begin{subfigure}[b]{0.3\linewidth}
                \includegraphics[width=\textwidth]{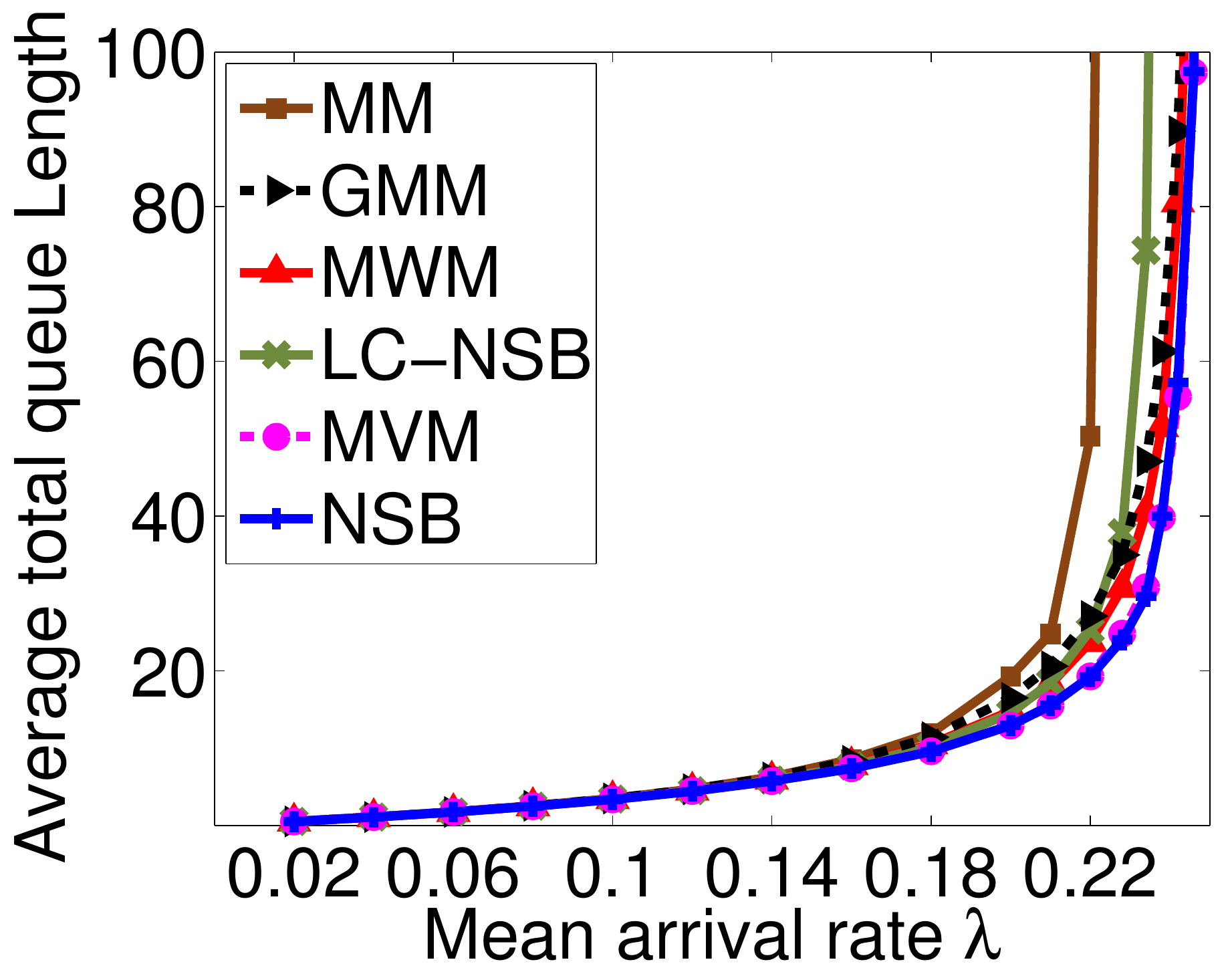}
                \caption{4$\times$4 grid}
                \label{fig:grid_result}
        \end{subfigure}
	  \quad
	    \begin{subfigure}[b]{0.3\linewidth}
                \includegraphics[width=\textwidth]{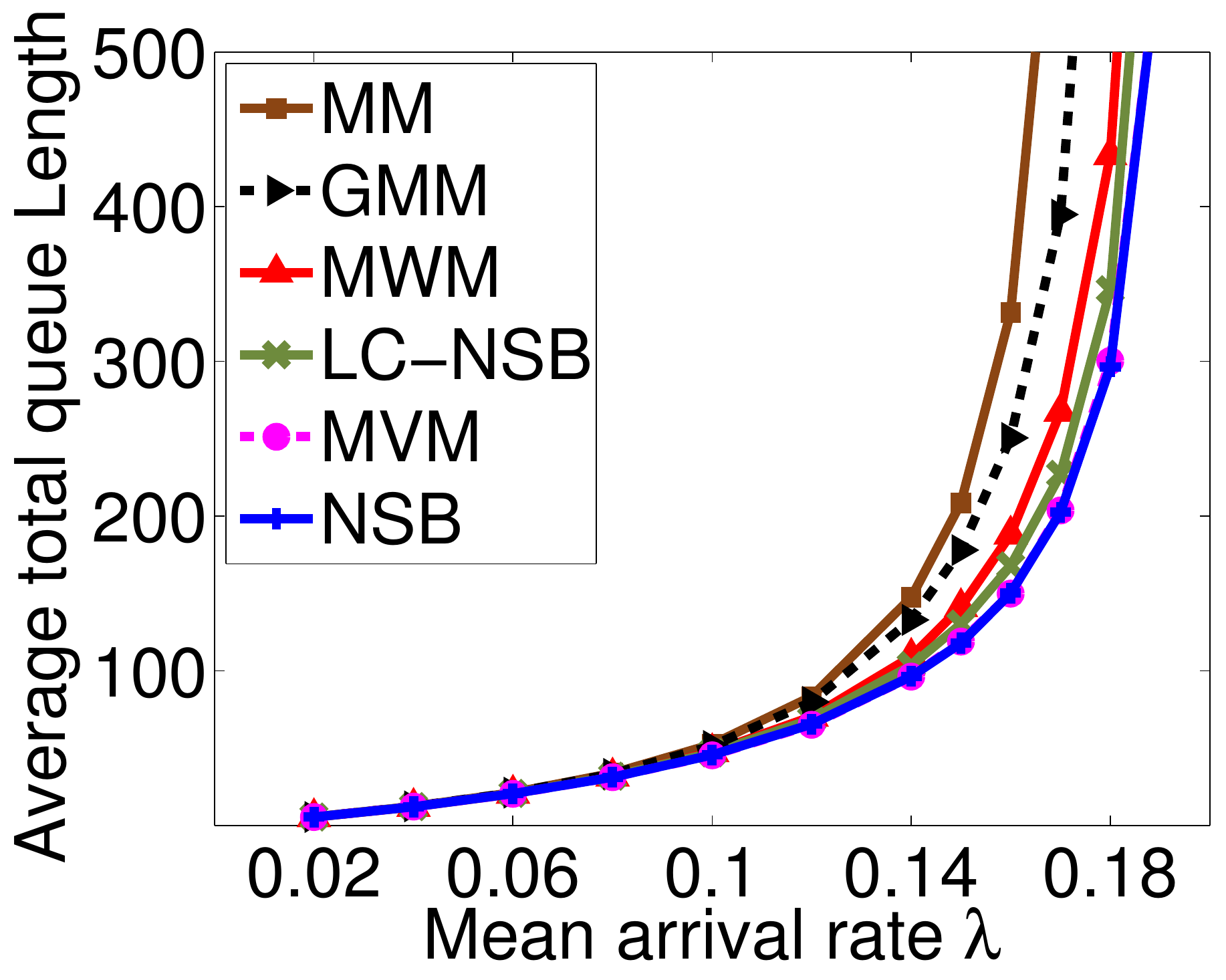}
                \caption{Random topology}
                \label{fig:random_result}
        \end{subfigure}  
        \caption{Simulation results for Poisson arrivals.}
        \label{fig:results_poi}
\end{figure*}

\begin{figure*}[t!]
        \centering
        \begin{subfigure}[b]{0.3\linewidth}
                \includegraphics[width=\textwidth]{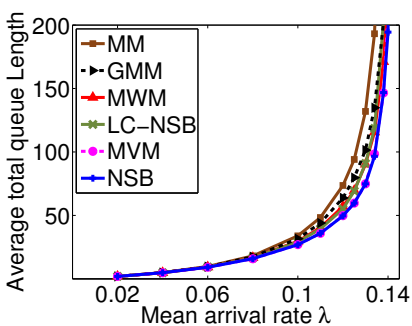}
                \caption{Triangular Mesh}
                \label{fig:triangle_result_file}
        \end{subfigure}%
        \quad
        \begin{subfigure}[b]{0.3\linewidth}
                \includegraphics[width=\textwidth]{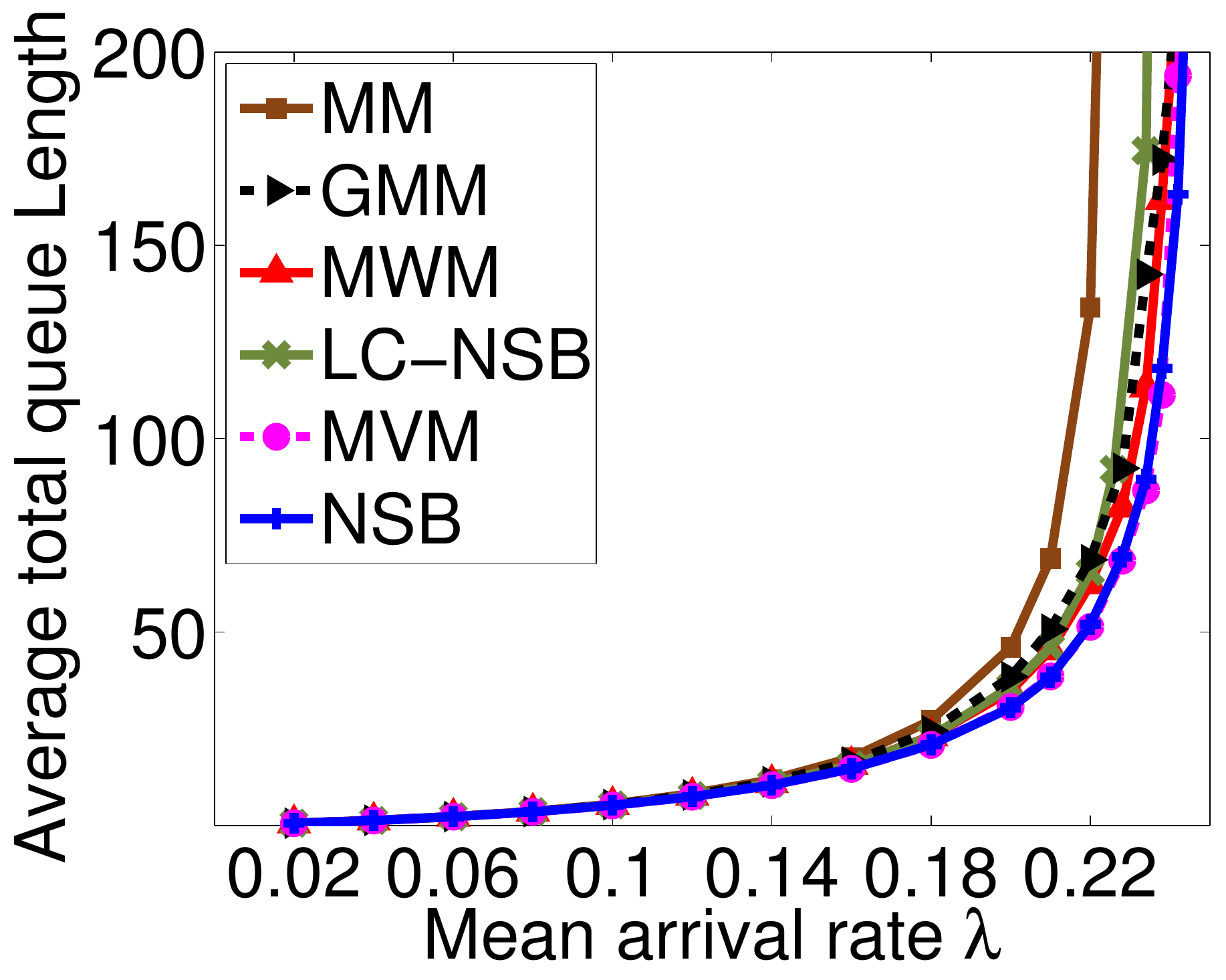}
                \caption{4$\times$4 grid}
                \label{fig:grid_result_file}
        \end{subfigure}
	  \quad
	    \begin{subfigure}[b]{0.3\linewidth}
                \includegraphics[width=\textwidth]{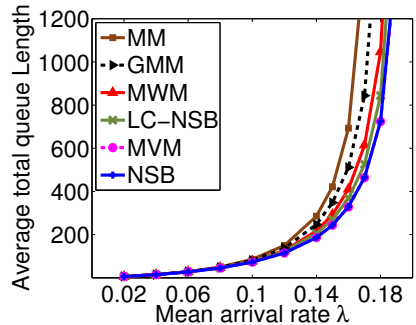}
                \caption{Random topology}
                \label{fig:random_result_file}
        \end{subfigure}  
        \caption{\HIGH{Simulation results for file arrivals, where the file arrival probability is $p=0.1$ and the file size follows Poisson distribution with mean $\lambda/p$.}}
        \label{fig:results_file}
\end{figure*}

\begin{figure*}[t!]
        \centering
        \begin{subfigure}[b]{0.3\linewidth}
                \includegraphics[width=\textwidth]{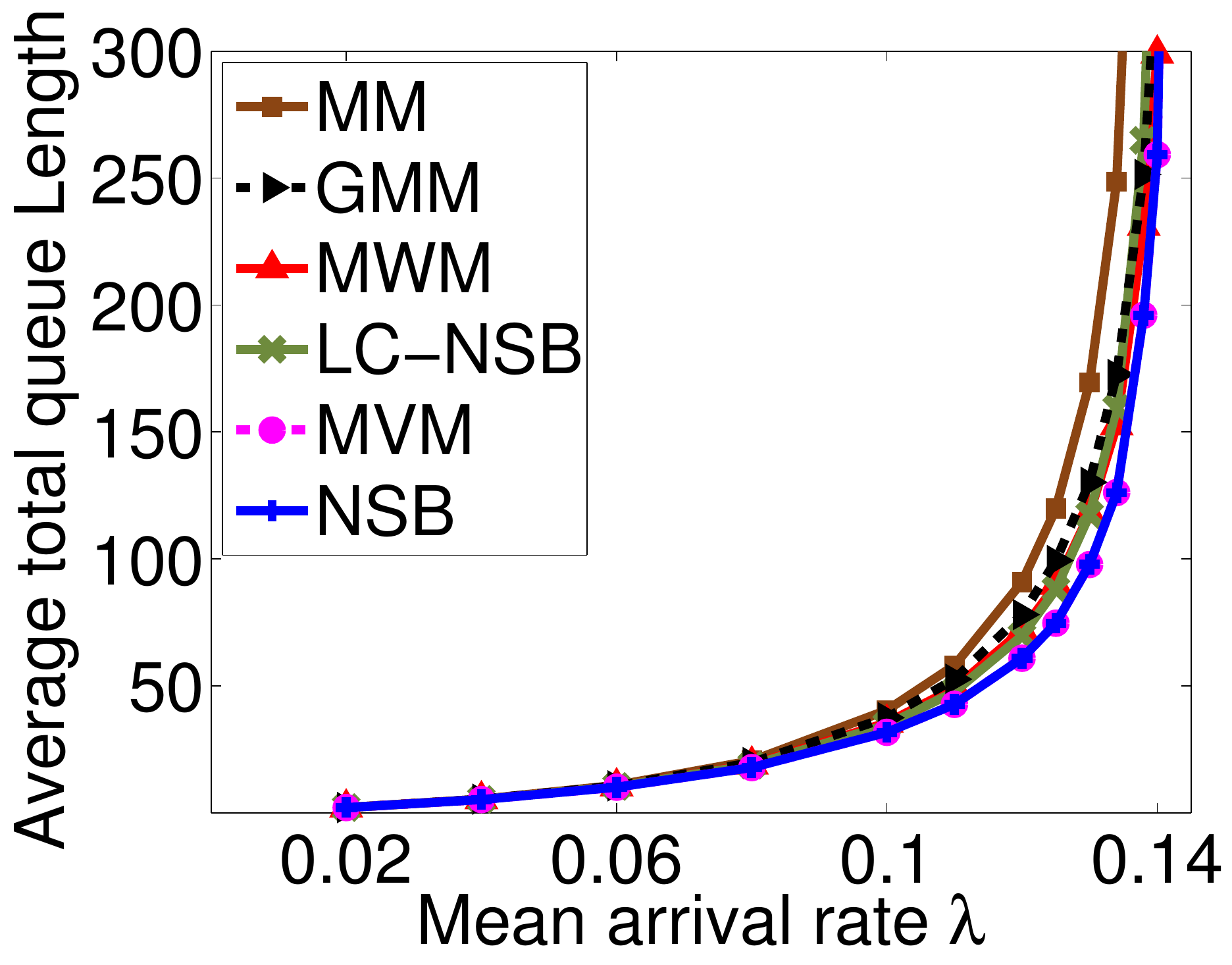}
                \caption{Triangular Mesh}
                \label{fig:triangle_result_zipf}
        \end{subfigure}%
        \quad
        \begin{subfigure}[b]{0.3\linewidth}
                \includegraphics[width=\textwidth]{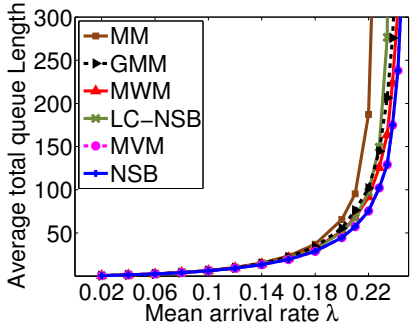}
                \caption{4$\times$4 grid}
                \label{fig:grid_result_zipf}
        \end{subfigure}
	  \quad
	    \begin{subfigure}[b]{0.3\linewidth}
                \includegraphics[width=\textwidth]{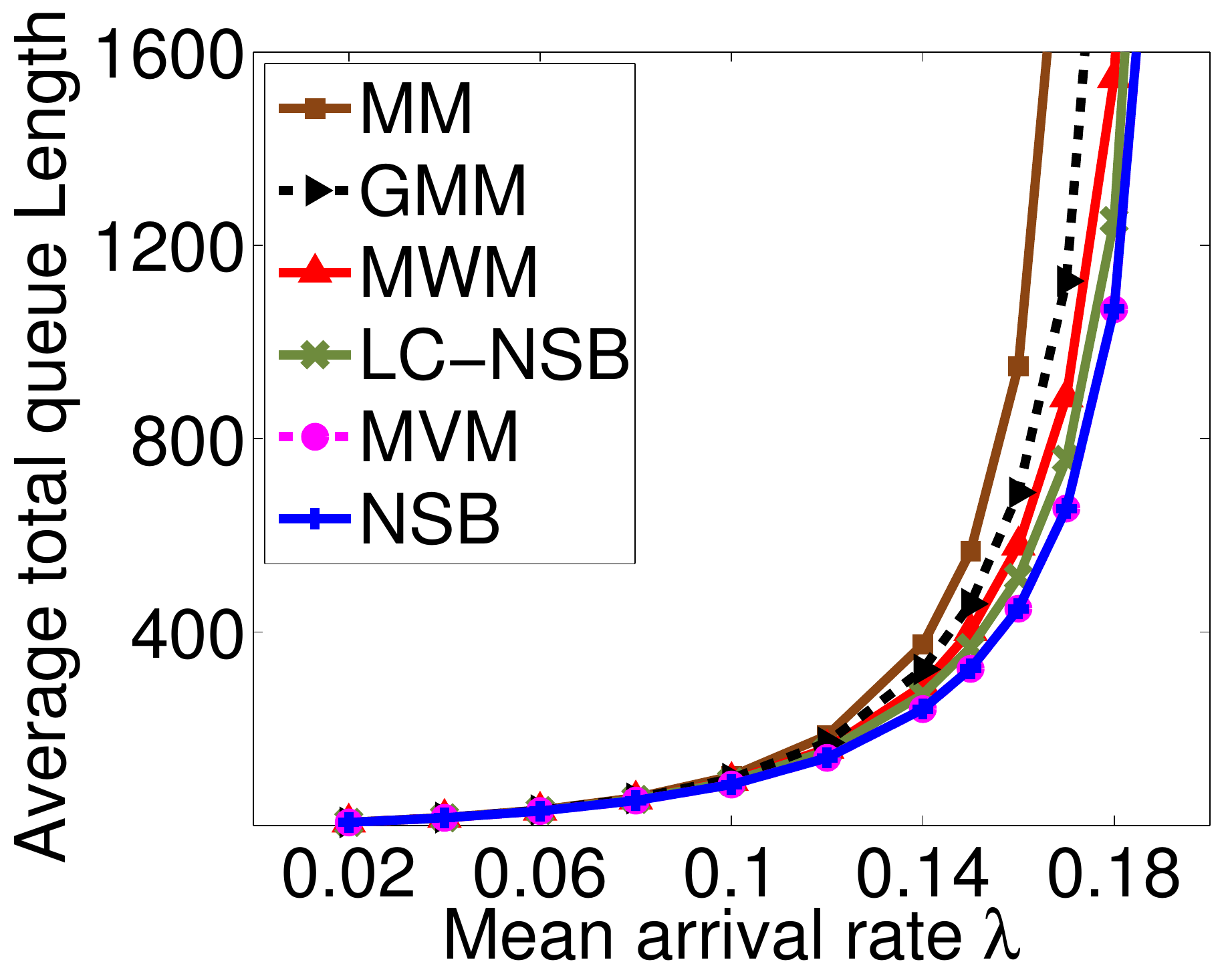}
                \caption{Random topology}
                \label{fig:random_result_zipf}
        \end{subfigure}  
        \caption{\HIGH{Simulation results for Zipf arrivals with a support of $[0,1,\dots,999]$.}}
        \label{fig:results_zipf}
\end{figure*}

\section{A Lower-Complexity NSB Algorithm} \label{sec:lc-nsb}
Through the analysis for the NSB algorithm, we obtain the following important insights:
In order to achieve the same performance guarantees as NSB, what really matters is the priority or 
the ranking of the nodes, rather than the exact weight of the nodes. 
\high{This insight comes from the following:
Note that under the NSB algorithm, the weight of each node is only used in the MVM component
(line 6 of Algorithm~\ref{alg:nsb}). In the performance analysis of NSB (i.e., Theorems~\ref{thm:nsb_et}, 
\ref{thm:nsb_throughput}, and \ref{thm:nsb_opt}), the only property of MVM we use is Lemma~\ref{lem:path}, 
which is concerned about the $s$ heaviest nodes (i.e., the $s$ highest ranking nodes) rather than 
about the exact weight of the nodes.}
Hence, if we assign the node weights in a way such that, the weights are bounded integers and the nodes 
still have the desired priority or ranking as in the NSB algorithm, then we can develop a new algorithm 
with a lower complexity. Thanks to the results of 
\cite{huang12,pettie12}, an $O(m \sqrt{n})$-complexity implementation\footnote{This can be done 
by setting the weight of an edge to the sum of the weight of its two end nodes and finding an MWM 
based on the new edge weights using the techniques developed in \cite{huang12,pettie12}.} of MVM 
can be derived if the maximum node weight is a bounded integer independent of $n$ and $m$. 

Next, we propose such an algorithm, called the Lower-Complexity NSB (LC-NSB). 
\high{Similarly as in NSB, we consider frames each consisting of three consecutive time-slots.}
Recall that $U_i(k)$ indicates whether node $i$ was matched in the previous time-slot (or in both of the previous two time-slots),
as defined in Eq.~(\ref{eq:U}). Also, recall that $\mC(k)$ and $\mH(k)$ denote the set of critical nodes and 
the set of heavy nodes in time-slot $k$, respectively. In time-slot $k$, we assign a weight to node $i$ as
\begin{equation}
\label{eq:weight_lc}
w_i(k) \triangleq
\begin{cases}
5-2U_i(k) &\text{if $i \in \mC(k)$}; \\
4-2U_i(k) &\text{if $i \in \mH(k) \backslash \mC(k)$}; \\
1  &\text{otherwise}.
\end{cases} 
\end{equation}

Then, the LC-NSB algorithm finds an MVM based on the assigned node weight $w_i(k)$'s in every time-slot. 
Note that LC-NSB has a very similar way of assigning the node weights as NSB. However, the key difference 
is that we now divide all the nodes into five priority groups by assigning the node weights only based on 
whether it is a heavy (or critical) node and whether it was scheduled in the previous time-slot(s), while in 
the NSB algorithm, the actual workload is used in the weight assignments. This slight yet crucial change 
leads to a lower-complexity algorithm with the same performance guarantees. Note that in Eq.~(\ref{eq:weight_lc}), 
we give a higher priority to the critical nodes in order to guarantee the evacuation time performance.
\HIGH{The proof follows a similar line of analysis to that for the NSB algorithm and is provided in the appendix for completeness.}

\begin{theorem}
\label{thm:lc_nsb}
The LC-NSB algorithm has an approximation ratio no greater than $3/2$ for the evacuation time and has an 
efficiency ratio no smaller than $2/3$ for the throughput. 
Moreover, the LC-NSB algorithm is both throughput-optimal and evacuation-time-optimal in bipartite graphs.
\end{theorem}


\emph{Remark:} Although the LC-NSB algorithm can provide the same performance guarantees as NSB, 
we would expect that LC-NSB may have (slightly) worse empirical performance compared to NSB, since 
NSB has a more fine-grained priority differentiation among all the nodes. We indeed make such observations 
in our simulation results in Section~\ref{sec:sim}. In order to improve the empirical performance, we can 
introduce more priority groups for the non-heavy nodes under LC-NSB rather than all being in the same 
priority group (of weight 1 as in Eq.~(\ref{eq:weight_lc})). As long as the number of priority groups is 
a bounded integer independent of $n$ and $m$, the complexity remains $O(m \sqrt{n})$.

\section{Numerical Results} \label{sec:sim}
In this section, we conduct numerical experiments to elucidate our theoretical results. We also compare the empirical 
performance of our proposed NSB and LC-NSB algorithms with several most relevant algorithms as listed in Table~\ref{tab:com}.

\subsection{Throughput Performance}\label{subsec:throughput}
\high{To evaluate the throughput performance, we run the simulations on three different network topologies as shown in Fig.~\ref{fig:topologies}.}
We first focus on a randomly generated triangular mesh topology with 30 nodes and 79 links as shown in Fig.~\ref{fig:triangle_top}. 
The simulations are implemented using C++. We assume that the arrivals are \emph{i.i.d.} over all the links with unit capacity. 
The mean arrival rate of each link is $\lambda$, and the instantaneous arrivals to each link follow a Poisson distribution in each 
time-slot. In Fig.~\ref{fig:triangle_result}, we plot the average total queue length in the system against the arrival rate $\lambda$. 
We consider several values of $\lambda$ as indicated in Fig.~\ref{fig:triangle_result}. For each value of $\lambda$, the average 
total queue length is an average of 10 independent simulations. Each individual simulation runs for a period of $10^5$ time-slots. 
We compute the average total queue length by excluding the first $5 \times 10^4$ time-slots in order to remove the impact of the 
initial transient state. Note that this network topology contains odd-size cycles. Hence, our proposed NSB and LC-NSB only 
guarantee to achieve $2/3$ of the optimal throughput. However, the simulation results in Fig.~\ref{fig:triangle_result} show 
that NSB and LC-NSB algorithms both empirically achieve the optimal throughput performance. 
\HIGH{
This is because the odd-size cycles (i.e., all the triangles in this case) do not form the bottlenecks in this setting. 
For example, a triangle requires $\lambda \le 1/3$ because at most one of its three links can be scheduled in each time-slot.
However, in Fig.~\ref{fig:triangle_top} there exists a node touched by seven links, which requires $\lambda \le 1/7$. 
As the load $\lambda$ increases, such a node will become congested sooner than any triangle and thus forms a scheduling bottleneck.
}

\high{We also conduct simulations for a 4$\times$4 grid topology with 16 nodes and 24 links (Fig.~\ref{fig:grid_top}) 
and a randomly generated dense topology with 100 nodes and 248 links (Fig.~\ref{fig:random_top}).
All the other simulation settings are the same as that for the triangular mesh topology.  
The simulation results are presented in Fig.~\ref{fig:results_poi}.
}
The observations we make are similar to that for the triangular 
mesh topology, except that in the grid topology, LC-NSB has higher delays when the mean arrival rate approaches 
the boundary of the optimal throughput region. This is due to the following two reasons: 1) Under LC-NSB, all the 
non-heavy nodes are in the same priority group of weight 1, so a non-heavy node with a large workload would have
a similar chance of getting scheduled as a non-heavy node with a small workload; 2) In the grid topology, the bottleneck
nodes (i.e., the four nodes in the center) are all adjacent. Note that in most cases, at least one of the bottleneck nodes
is the critical node and will have the highest priority. If another bottleneck node is adjacent to the critical node but is not
a heavy node, this bottleneck node will get a lower chance of being scheduled, since the critical node could be matched
with other node. This inefficiency does not occur under NSB, since the node weights of the non-heavy nodes are their
workload and are more fine-grained. Hence, a non-heavy node with a large workload will have a higher priority than a 
non-heavy node with a small workload.

\HIGH{
In order to evaluate the throughput performance under different arrival patterns,
we also run simulations for file arrivals. Specifically, the arrivals have the following pattern: in each time-slot, there 
is a file arrival with probability $p$, and no file arrival otherwise; the file size follows Poisson distribution with mean $\lambda/p$.
The simulation results for $p=0.1$ are presented in Fig.~\ref{fig:results_file}.
Similar observations to that for Poisson arrivals can be made, except that all the algorithms have larger delays due to 
a more bursty arrival pattern. 
In addition, we also consider more realistic arrival patterns where the arrivals in each time-slot follow the Zipf law, 
which is commonly used to model the Internet traffic \cite{feldmann04}. We assume a support of \HIGH{$[0,1,\dots,999]$} 
for the Zipf distribution. 
The power exponent of the Zipf distribution is determined based on the mean arrival rate $\lambda$. 
The simulation results are presented in Fig.~\ref{fig:results_zipf}. 
The overall observations are again similar to that for the previous arrival patterns.

Finally, it is remarkable that in all simulation settings we consider, our proposed node-based algorithm NSB empirically 
achieves the best delay performance. When the traffic load is high, NSB even results in a significant reduction (10\%-30\%) 
in the average delay performance compared to the link-based algorithms such as MWM (e.g., in Fig.~\ref{fig:random_result_zipf} 
for a random topology with Zipf arrivals, the delay reduction is about 30\% when $\lambda=0.18$).
Although NSB ties with another node-based algorithm MVM for the empirical delay performance, as we discussed in the introduction, 
the throughput performance of MVM is not well understood yet.
}

\high{
\subsection{Evacuation Time Performance}
In this subsection, we evaluate the evacuation time performance of our proposed algorithms.
As we discussed in the introduction, the minimum evacuation time problem is equivalent to 
the classic multigraph edge coloring problem, for which there are common benchmarks.
Therefore, we run simulations for six DIMACS benchmark instances~\cite{dimacs} for the graph coloring problem. 
In addition, we also run simulations for three regular multigraphs, \HIGH{three random multigraphs,} and one special graph 
(Fig.~\ref{fig:graphN} with $N=100$).

The evacuation time performance for each of the considered algorithms under each graph is presented 
in Table~\ref{tab:benchmark}. 
The simulation results show that all the algorithms have very similar evacuation time performance for the 
considered benchmark graphs and regular multigraphs, although they have different theoretical guarantees.
For the special graph in Fig.~\ref{fig:graphN}, we can observe that the node-based algorithms exhibit a much 
better evacuation time performance compared to the link-based algorithms \HIGH{(e.g., MWM and GMM)}. 
Specifically, the link-based algorithms require about twice as much time as that of the node-based algorithms 
to evacuate all initial packets in the network.
}


\section{Conclusion} \label{sec:con}
In this paper, we studied the link scheduling problem for multi-hop wireless networks 
and focused on designing efficient online algorithms with provably guaranteed throughput and evacuation 
time performance. We developed two node-based service-balanced algorithms and showed that none of 
the existing algorithms strike a more balanced performance guarantees than our proposed algorithms in 
both dimensions of throughput and evacuation time. 
\HIGH{An important future direction is to consider more general models (which, e.g., allow for multi-hop traffic, 
general interference models, and time-varying channels). In such scenarios, it becomes much more challenging 
to provide provably good evacuation time performance.}

\begin{table}[t!]
\centering
\begin{tabular}{c|c c c c c c }
\hline
Graph & MWM  & GMM & MVM & NSB & LC-NSB & $\Delta$ \rule{0pt}{2ex}\rule[-1ex]{0pt}{0pt}\\ 
\hline
dsjc125.1 &	23&	23&	23&	23&	23&	23 \\
dsjc125.5 &	76&	75&	75&	75&	75	&	75\\
dsjc125.9 &	140&	120&	120&	120&	120 &	120\\
dsjc250.1 &	38&	40&	38&	38&	38	&	38\\
dsjc250.5 &	147&	153&	147&	147&	147 &	147\\
dsjc250.9 &	234&	234&	234&	234& 234&	234\\
\hline
regm50.20 &	20&	 25&	 20&	 20&	 20&	 20 \\
regm50.50 & 	50&	 55&	 51&	 51&	 51&	 50 \\
regm50.80 &	80&	 84&  80&	 80&	 80&	 80 \\
\hline
rand100.50 &	366&	 366& 366&  366&  366&  366 \\
rand100.100 & 	813&	 813&  813&  813&  813&  813 \\
rand100.250 &	2161&  2161&  2161&  2161&  2161&  2161 \\
\hline
Fig.~\ref{fig:graphN} &  199& 199& 101& 101& 101& 101 \\
\hline 
\end{tabular}
\centering
\caption{Evacuation time performance for six DIMACS benchmark graphs~\cite{dimacs}, three regular multigraphs,
and one special graph (Fig.~\ref{fig:graphN} with $N=100$). In the table, $\Delta$ denotes the maximum node degree,
dsjcX.Y denotes the label of the DIMACS benchmark graphs, regmX.Y denotes a regular mutigraph with X nodes and 
node degree Y, \HIGH{and randX.Y denotes a random topology in Fig.~\ref{fig:random_top} with X nodes and the number 
of multi-edges at each link being uniformly distributed over the interval [0,Y].}}
\label{tab:benchmark}
\end{table}

\section{Proofs} \label{sec:proofs}

\subsection{Proof of Proposition~\ref{pro:nsb}}\label{sec:pro:nsb}
We first restate a useful result of \cite{anstee96} in Lemma~\ref{lem:delta}, which will be used 
in the proof of Proposition~\ref{pro:nsb}. \high{Throughout the paper, we assume that the multigraph
is loopless (i.e., there is no edge connecting a node to itself) unless explicitly mentioned.}

\begin{lemma}[Theorem~1 of \cite{anstee96}]
\label{lem:delta}
Let $\Graph$ be a loopless multigraph with maximum degree $\Delta$. Let $G_{\Delta}$ denote 
the subgraph of $\Graph$ \high{induced\footnote{\high{An induced subgraph of a graph is formed from 
a subset of the nodes of the graph and all of the edges whose endpoints are both in this subset.}}}
by all the nodes having maximum degree. If $\Graph_{\Delta}$ 
is bipartite, then there exists a matching over $\Graph$ that matches every node of maximum degree.
\end{lemma}

Now, we are ready to prove Proposition~\ref{pro:nsb}. 

\begin{proof}[Proof of Proposition~\ref{pro:nsb}]
First, note that in any time-slot, the network together with the present packets can be 
represented as a loopless multigraph.
Recall that $\Graph(k)$ denotes the multigraph at the beginning of time-slot $k$ and 
$M(k)$ denotes the matching found by the NSB algorithm in time-slot $k$. Also, recall 
that the degree of node $i$ in $\Graph(k)$ is equivalent to the node queue length 
$Q_i(k)$, and the maximum node degree of $\Graph(k)$ is equal to $\Delta(k)$.
\high{Now, consider any frame $k^{\prime}$ consisting of three consecutive time-slots 
$\{p, p+1, p+2\}$, where $p=3k^{\prime}$. 
Suppose that the maximum node queue length is no smaller than two at the beginning of 
frame $k^{\prime}$, i.e., $\Delta(p) \ge 2$ at the beginning of time-slot $p$.}
Then, we want to show that under the NSB algorithm, the maximum degree will be 
at most $\Delta(p)-2$ at the end of time-slot $p+2$.
We proceed the proof in two steps: 
1) we first show that the maximum degree will decrease by at least one in the first 
two time-slots $p$ and $p+1$ (i.e., the maximum degree will be at most $\Delta(p)-1$
at the end of time-slot $p+1$),
and then, 2) show that if the maximum degree decreases by exactly one in the first 
two time-slots (i.e., the maximum degree is $\Delta(p)-1$ at the end of time-slot $p+1$), 
then the maximum degree must decrease by one in time-slot $p+2$, and becomes 
$\Delta(p)-2$ at the end of time-slot $p+2$. 

We start with step 1). It is a trivial case if the maximum degree decreases by one in time-slot $p$. 
Therefore, suppose the maximum degree does not decrease in time-slot $p$. Then, it suffices to 
show that all the nodes having maximum degree $\Delta(p)$ in $\Graph(p+1)$ must be scheduled in 
time-slot $p+1$ under the NSB algorithm. Note that matching $M(k)$ must be a maximal matching 
over $\Graph(k)$ for every time-slot $k$. Since $M(p)$ is a maximal matching, the nodes having 
maximum degree must form an independent set over $\Graph(p+1)$ at the beginning of time-slot 
$p+1$. We prove this by contradiction. 
\high{Note that if there is only one node having maximum degree at the beginning of time-slot $p+1$, 
then it is trivial that the subgraph induced by this single node must consist of this node itself only and 
thus forms an independent set.} 
So we consider the case where there are at 
least two nodes having maximum degree at the beginning of time-slot $p+1$. Suppose node $i$ 
and node $j$ are two adjacent nodes having maximum degree $\Delta(p)$ at the beginning of time-slot 
$p+1$. Then, none of the edges incident to either $i$ or $j$ was in matching $M(p)$. This implies 
that the edge between $i$ and $j$ can be added to matching $M(p)$ in time-slot $p$, which, however, 
contradicts the fact that $M(p)$ is a maximal matching. Therefore, the nodes having maximum degree 
must form an independent set at the beginning of time-slot $p+1$. 
\high{Clearly, the subgraph induced by all the nodes having maximum degree forms an independent set 
and thus has no edges. In this case, it is trivial that this induced subgraph is bipartite}. 
Then, by Lemma~\ref{lem:delta}, there exists a matching \high{over $\Graph(p+1)$} that matches all the 
nodes having maximum degree 
in time-slot $p+1$. Note that $M(p+1)$ is an MVM over $\Graph(p+1)$ with the assigned node weights 
(as in Eq.~(\ref{eq:weight})) under the NSB algorithm. It is also easy to see that all the nodes with 
maximum degree $\Delta(p)$ are among the ones with the heaviest weight, as they have a weight of 
$2 \Delta(p)$ and the weight of all the other nodes is less than $2 \Delta(p)$. Hence, it implies from 
Lemma~\ref{lem:path} that matching $M(p+1)$ also matches all the nodes having maximum degree, 
i.e., the maximum degree decreases by one in time-slot $p+1$. This completes the proof of step 1).

Now, we prove step 2).
Clearly, the maximum degree becomes $\Delta(p)-1$ at the beginning of time-slot $p+2$. Recall
that $\mC(p+2)$ denotes the set of critical nodes. We want to show that all the nodes in 
$\mC(p+2)$ will be matched in time-slot $p+2$. We first show that all the nodes in $\mC(p+2)$ 
are among the ones with the heaviest weights at the beginning of time-slot $p+2$. This is true 
due to the following. It is easy to see that for any node $i \in \mC(p+2)$, it was matched at most 
once in time-slots $p$ and $p+1$. Hence, according to the weight assignments in Eq. (\ref{eq:weight}), 
node $i$ has a weight of $2(\Delta(p)-1)$, while all the nodes in $\Vertex \backslash \mC(p+2)$ must 
have a degree less than $\Delta(p)-1$ and thus have a weight less than $2(\Delta(p)-1)$. Therefore, 
all the nodes in $\mC(p+2)$ are among the ones with the heaviest weights.

Let $\Graph_{\mC(p+2)}$ denote the subgraph of $\Graph(p+2)$ induced by all the nodes in 
$\mC(p+2)$. If $\Graph_{\mC(p+2)}$ is bipartite, then again by Lemmas~\ref{lem:path} 
and \ref{lem:delta}, following the same argument as in step 1), we can show that matching 
$M(p+2)$ matches all the nodes in $\Graph_{\mC(p+2)}$ in time-slot $p+2$. Therefore, it 
remains to show that $\Graph_{\mC(p+2)}$ is bipartite. We prove this by contradiction. 
Suppose $\Graph_{\mC(p+2)}$ contains an odd cycle, say $C$. Then, no two adjacent nodes 
of $C$ were matched by $M(p+1)$ in time-slot $p+1$. This is true due to the following. 
Suppose there exist two adjacent nodes of $C$, say $i$ and $j$, matched by $M(p+1)$. 
Since $i$ and $j$ are in $\mC(p+2)$ (i.e., their degree is $\Delta(p)-1$ in time-slot $p+2$), then
they both have maximum degree $\Delta(p)$ at the beginning of time-slot $p+1$. However, given 
that $i$ and $j$ are adjacent in $\Graph(p+2)$ (and are thus adjacent in $\Graph(p+1)$ as well), 
this contradicts what we have shown earlier -- the nodes of degree $\Delta(p)$ in $\Graph(p+1)$ 
form an independent set. Therefore, no two adjacent nodes of $C$ were matched by $M(p+1)$ 
in time-slot $p+1$. This, along with the fact that cycle $C$ is of odd size, implies that cycle $C$ 
must contain two adjacent nodes that were not matched by $M(p+1)$ in time-slot $p+1$. This 
further implies that the edge between these two adjacent nodes can be added to $M(p+1)$, 
which contradicts the fact that $M(p+1)$ is a maximal matching over $\Graph(p+1)$. Therefore, 
the induced subgraph $\Graph_{\mC(p+2)}$ must be bipartite. This completes the proof of step 2), 
as well as the proof of Proposition~\ref{pro:nsb}.
\end{proof}

\subsection{Proof of Lemma~\ref{lem:2in3}} \label{sec:lem:2in3}
%
\begin{proof}
Recall that $\mC$ is the set of critical nodes in the fluid limits at scaled time $t$ (Eq.~(\ref{eq:critical})).
We want to show that under the NSB algorithm, all the nodes in $\mC$ will be scheduled at least 
twice within \high{each frame of interval $T$}.

First, recall that $j$ is large enough such that 
$\left \lvert Q_i(\xrj \tau)/\xrj - q_i(\tau) \right \rvert < \beta/2$ for all times
$\tau \in (t,t+\delta)$. Hence, from condition (C1) and (C2), the queue lengths 
in the original system satisfy the following conditions for all time-slots $k \in T$: 

(C1*) $Q_i(k) \in \xrj (q_{\max}(t)-\beta, q_{\max}(t)+\beta)$ for all $i \in \mC$;

(C2*) $Q_i(k) < \xrj (q_{\max}(t)-2\beta)$ for all $i \notin \mC$.

\noindent On account of condition (C1*) and Eq.~(\ref{eq:heavy}), all the 
nodes in $\mC$ are heavy nodes in all the time-slots of $T$, i.e., 
$Q_i(k) \ge (n-1)/n \cdot \Delta(k)$ for all $i \in \mC$ and for all $k \in T$.

Note that in any time-slot, the network together with the present packets can be mapped to a 
multigraph, where each multi-edge corresponds to a packet. Recall that we use $\Graph(k)$ to 
denote the multigraph at the beginning of time-slot $k$. Note that if there are no packets waiting 
to be transmitted over a link, no multi-edge connecting the end nodes of this link will appear in 
$\Graph(k)$. Also, recall that $M(k)$ denotes the matching found by the NSB algorithm in time-slot $k$. 
\high{Now, consider any frame $k^{\prime}$ of interval $T$ consisting of three consecutive time-slots 
$\{p, p+1, p+2\}$, where $p=3k^{\prime}$.} 
We want to show that under the NSB algorithm, every node in 
$\mC$ will get scheduled in at least two time-slots of $\{p, p+1, p+2\}$. We proceed the proof in 
two steps: 1) we first show that all the nodes in $\mC$ will be scheduled at least once in the first 
two time-slots $p$ and $p+1$, and 2) then show that all the nodes in $\mC$ that were scheduled 
exactly once in the first two time-slots, will get scheduled in time-slot $p+2$. 

We start with step 1). Let $\mC^{\prime}$ denote the set of nodes in $\mC$ that were not 
scheduled in time-slot $p$. It is a trivial case if $\mC^{\prime} = \emptyset$. Therefore, 
suppose $\mC^{\prime} \neq \emptyset$, i.e., there exists at least one node in $\mC$ that 
was not scheduled in time-slot $p$. Then, it suffices to show that all the nodes in $\mC^{\prime}$ 
must be scheduled in time-slot $p+1$ under the NSB algorithm. Note that matching $M(k)$ 
must be a maximal matching over $\Graph(k)$ for every time-slot $k$. Since $M(p)$ is 
a maximal matching, the nodes in $\mC^{\prime}$ must form an independent set at the 
beginning of time-slot $p+1$, excluding the multi-edges corresponding to the new packet 
arrivals at the beginning of time-slot $p+1$. 

Note that it is a trivial case if $\lvert \mC^{\prime} \rvert=1$. So we consider the case of 
$\lvert \mC^{\prime} \rvert \ge 2$ and prove it by contradiction. Suppose there exist two 
adjacent nodes $i, j \in \mC^{\prime}$. Then, none of the edges incident to either $i$ or 
$j$ was in matching $M(p)$. This implies that the multi-edge between $i$ and $j$ could 
be added to matching $M(p)$ in time-slot $p$, which, however, contradicts the fact that 
$M(p)$ is a maximal matching. Therefore, the nodes in $\mC^{\prime}$ must form an 
independent set at the beginning of time-slot $p+1$. 
\high{Clearly, the subgraph induced by all the nodes in $\mC^{\prime}$ forms an independent set 
and thus has no edges. In this case, it is trivial that this induced subgraph is bipartite.} 
Note that conditions (C1*) and (C2*) still hold even without accounting for the new packet arrivals. 
Then, by Lemma~\ref{lem:existence}, there exists a matching that matches all the nodes 
in $\mC^{\prime}$ at the beginning of time-slot $p+1$ before new packet arrivals. Clearly, 
such a matching still exists even if the multi-edges corresponding to the newly arrived 
packets in time-slot $p+1$ are added to the grpah. Note that $M(p+1)$ is an MVM over 
$\Graph(p+1)$ with the assigned weights (as in Eq.~(\ref{eq:weight})). Now, if all the 
nodes in $\mC^{\prime}$ are among the ones with the heaviest weights, then it implies 
from Lemma~\ref{lem:path} that matching $M(p+1)$ also matches all the nodes in 
$\mC^{\prime}$. This is indeed true due to conditions (C1*) and (C2*), as well as the 
weight assignments in Eq.~(\ref{eq:weight}): every node in $\mC^{\prime}$ was not 
scheduled in time-slot $p$, and thus has a weight larger than $2\xrj(q_{\max}(t)-\beta)$, 
while any node in $\Vertex \backslash \mC^{\prime}$ cannot have a weight larger than 
$\max \{2\xrj(q_{\max}(t)-2\beta), \xrj(q_{\max}(t)+\beta)\}$.

Now, we prove step 2).
Let $\mC^{\prime \prime}$ denote the set of nodes in $\mC$ that were scheduled exactly 
once in time-slots $p$ and $p+1$. We want to show that all the nodes in $\mC^{\prime \prime}$ 
will get scheduled in time-slot $p+2$. Note that all the nodes in $\mC^{\prime \prime}$ are 
among the ones with the heaviest weights. This is true due to conditions (C1*) and (C2*), 
as well as the weight assignments in Eq.~(\ref{eq:weight}): every node in $\mC^{\prime \prime}$ 
was scheduled exactly once in time-slots $p$ and $p+1$, and thus has a weight larger than 
$2\xrj(q_{\max}(t)-\beta)$, while any node in $\Vertex \backslash \mC^{\prime \prime}$ cannot 
have a weight larger than $\max \{2\xrj(q_{\max}(t)-2\beta), \xrj(q_{\max}(t)+\beta)\}$. Further, 
let $\Graph_{\mC^{\prime \prime}}$ denote the subgraph induced by all the nodes in 
$\mC^{\prime \prime}$ at the beginning of time-slot $p+2$, excluding all the multi-edges 
corresponding to the packets that arrived in time-slot $p+1$ and $p+2$. 
If $\Graph_{\mC^{\prime \prime}}$ is bipartite, then again by Lemmas~\ref{lem:existence} 
and \ref{lem:path}, following the same argument as in step 1), we can show that all the 
nodes in $\mC^{\prime \prime}$ are matched by $M(p+2)$ in time-slot $p+2$. Therefore, 
it remains to show that $\Graph_{\mC^{\prime \prime}}$ is bipartite. 

Next, we prove that $\Graph_{\mC^{\prime \prime}}$ is bipartite by contradiction. Suppose 
$\Graph_{\mC^{\prime \prime}}$ contains an odd cycle, say $C$. Then, no two adjacent 
nodes of $C$ were matched by $M(p+1)$ in time-slot $p+1$. This is true due to the following. 
Suppose there exist two adjacent nodes of $C$, say $i$ and $j$, matched by $M(p+1)$. 
Since $i$ and $j$ are in $\mC^{\prime \prime}$, both of them were matched exactly once in 
time-slots $p$ and $p+1$ from the definition of $\mC^{\prime \prime}$. This implies that both 
$i$ and $j$ were not matched in time-slot $p$, i.e., we have $i,j \in \mC^{\prime}$. However,
given that $i$ and $j$ are adjacent, this contradicts what we have shown earlier -- the nodes 
in $\mC^{\prime}$ form an independent set. Therefore, no two adjacent nodes of $C$ were 
matched by $M(p+1)$ in time-slot $p+1$. This, along with the fact that cycle $C$ is of odd size, 
implies that cycle $C$ must contain two adjacent nodes that were not matched by $M(p+1)$ 
in time-slot $p+1$. This further implies that the multi-edge between these two adjacent nodes 
can be added to $M(p+1)$, which contradicts the fact that $M(p+1)$ is a maximal matching 
over $\Graph(p+1)$. Therefore, the induced subgraph $\Graph_{\mC^{\prime \prime}}$ must 
be bipartite. This completes the proof of step 2) and that of Lemma~\ref{lem:2in3}.
\end{proof}

\bibliographystyle{IEEEtran}
\bibliography{nsb}

\begin{IEEEbiography}{Yu Sang}
received his B.E. degree in Electronic Information Engineering from
University of Science and Technology of China in 2014. He is currently
a Ph.D. student in the CIS department at Temple University. His research
interests are in wireless networks and cloud computing systems.
\end{IEEEbiography}

\begin{IEEEbiography}{Gagan R. Gupta}
received the Bachelor of Technology degree in Computer Science and 
Engineering from the Indian Institute of Technology, Delhi, India in 
2005. He received the M.S. degree in Computer Science from University 
of Wisconsin, Madison in 2006. He received his Ph.D. degree from Purdue 
University, West Lafayette, Indiana in 2009 for his dissertation titled 
``delay efficient control policies for wireless networks." His research 
interests are in performance modeling and optimization of communication 
networks and parallel computing. He is currently working in the telecommunications 
industry.
\end{IEEEbiography}

\begin{IEEEbiography}{Bo Ji}(S'11-M'12)
received his B.E. and M.E. degrees in Information Science and Electronic 
Engineering from Zhejiang University, China in 2004 and 2006, respectively. 
He received his Ph.D. degree in Electrical and Computer Engineering from 
The Ohio State University in 2012. He is currently an assistant professor of 
the CIS department at Temple University, Philadelphia, PA, and is also a 
faculty member of the Center for Networked Computing (CNC). Prior to joining 
Temple University, he was a Senior Member of Technical Staff with AT\&T Labs, 
San Ramon, CA, from January 2013 to June 2014. His research interests are 
in the modeling, analysis, control, and optimization of complex networked systems, 
such as communication networks, information-update systems, cloud/datacenter 
networks, and cyber-physical systems.
Dr. Ji received National Science Foundation (NSF) CAREER Award and 
NSF CISE Research Initiation Initiative (CRII) Award both in 2017.
\end{IEEEbiography}

\clearpage

\appendices

\section{Existence of Fluid Limits (Eqs.~\eqref{eq:fluid_a}-\eqref{eq:fluid_h})} \label{sec:fluid}
\begin{proof}
The proof follows a similar line of analysis in \cite{dai00}. Consider a fixed sample path $\omega$. 
For notational simplicity, in the following proof we omit the dependency on $\omega$.
Recall that a sequence of functions $f^{(s)}(\cdot)$ converges to $f(\cdot)$ uniformly 
on compact (\emph{u.o.c.}) intervals if for every $t \ge 0$, it is satisfied that 
$\lim_{s \rightarrow \infty} \sup_{0 \le t^{\prime} \le t} | f^{(s)}(t^{\prime}) - f(t^{\prime}) | = 0$.

First, for each $x>0$, we define the following:
\[
\begin{split}
a_i^{(x)}(t) &\triangleq \frac{A_i(xt)}{x}, \\
q_i^{(x)}(t) &\triangleq \frac{Q_i(xt)}{x}, \\
d_i^{(x)}(t) &\triangleq \frac{D_i(xt)}{x}, \\
h_M^{(x)}(t) &\triangleq \frac{H_M(xt)}{x}.
\end{split}
\]
Note that in each time-slot, only one matching can be chosen, and each node is scheduled at
most once. Hence, both $\{h_M^{(x)}(t)\}$ and $\{d_i^{(x)}(t)\}$ are a uniformly bounded sequence 
of functions with bounded Lipschitz constant. Specifically, for any $x>0$ and $0 \le t_1 \le t_2 \le t$, 
the following is satisfied:
\[
\begin{split}
h_M^{(x)}(t_2) - h_M^{(x)}(t_1) &= \frac{H_M(xt_2) - H_M(xt_1)}{x} \\
& \le \frac{x(t_2 - t_1)}{x} = t_2 - t_1, 
\end{split}
\]
and similarly,
\[
d_i^{(x)}(t_2) - d_i^{(x)}(t_1) \le t_2 - t_1.
\]
Then, the Arzela-Ascoli Theorem (e.g., see \cite{resnick13}) implies that for any positive sequence $\xr \rightarrow \infty$, 
there exists a subsequence $\xrj$ with $\xrj \rightarrow \infty$ as $j \rightarrow \infty$ and continuous
functions $h_M(\cdot)$ and $d_i(\cdot)$ such that for any $t \ge 0$,
\begin{eqnarray}
&& \lim_{j \rightarrow \infty} \sup_{0 \le t^{\prime} \le t} \left | h_M^{(\xrj)}(t^{\prime}) - h_M(t^{\prime}) \right | = 0, \label{eq:uoc_h} \\
&& \lim_{j \rightarrow \infty} \sup_{0 \le t^{\prime} \le t} \left | d_i^{(\xrj)}(t^{\prime}) - d_i(t^{\prime}) \right | = 0. \label{eq:uoc_d}
\end{eqnarray}
Also, it is easy to see that 
\begin{equation}
\label{eq:uoc_a}
\lim_{x \rightarrow \infty} \sup_{0 \le t^{\prime} \le t} \left | a_i^{(x)}(t^{\prime}) - \lambda_i t^{\prime} \right | = 0 
\end{equation}
for all the sample paths that satisfy the SLLN assumption (i.e., Eq.~\eqref{eq:slln_l}).
By combining Eqs.~\eqref{eq:uoc_d} and \eqref{eq:uoc_a}, we have that, for almost all sample 
paths (i.e., those that satisfy the SLLN assumption) and for any positive sequence $\xr \rightarrow \infty$, 
there exists a subsequence $\xrj$ with $\xrj \rightarrow \infty$ as $j \rightarrow \infty$
and continuous function $q_i(\cdot)$ such that for any $t \ge 0$, 
\[
\lim_{j \rightarrow \infty} \sup_{0 \le t^{\prime} \le t} \left | q_i^{(\xrj)}(t^{\prime}) - q_i(t^{\prime}) \right | = 0,
\]
and Eqs.~\eqref{eq:uoc_h} and \eqref{eq:uoc_d} hold.
This completes the proof of the convergence in Eqs.~\eqref{eq:fluid_a}-\eqref{eq:fluid_h}.
\end{proof}

\section{Proof of Theorem~\ref{thm:nsb_opt}} \label{sec:thm:nsb_opt}
\begin{proof}
Suppose that the underlying network graph is bipartite.
We first show that NSB achieves evacuation time optimality. Recall that for a given network 
with initial packets waiting to be transmitted, $\Delta(0)$ denotes the maximum node degree at the 
very beginning, and $\ChromaticIndex$ denotes the minimum evacuation time. If the underlying network is bipartite, 
there are no odd-size cycles, and we have $\ChromaticIndex=\Delta(0)$ (e.g., see \cite{konig16,kapoor00,guptathesis}). 
Hence, if the maximum degree decreases by one in every time-slot, the minimum evacuation time can be achieved. 
Therefore, we want to show that in every time-slot, all the critical nodes are matched under NSB.

Consider any time-slot $k$. Since the underlying network graph is bipartite, it is obvious that the subgraph
induced by all the heavy nodes is also bipartite. By Lemma~\ref{lem:existence}, we have that there exists 
a matching that matches all the heavy nodes. Due to the weight assignment rule in Eq.~\eqref{eq:weight},
we have $w_i(k) > w_{i^\prime}(k)$ for any $i \in \mH(k)$ and any $i^{\prime} \notin \mH(k)$. This is because
the weight of any heavy node is equal to either its queue length or twice its queue length, which is at least 
$(n-1)/n \cdot \Delta(k)$, i.e., $w_i(k) \ge (n-1)/n \cdot \Delta(k)$ for any $i \in \mH(k)$, and the weight of 
any non-heavy node is equal to its queue length, which is less than $(n-1)/n \cdot \Delta(k)$, i.e., 
$w_{i^{\prime}}(k) < (n-1)/n \cdot \Delta(k)$ for any $i^{\prime} \notin \mH(k)$. 
Then, due to Lemma~\ref{lem:path}, NSB matches all the heavy nodes, including all the critical nodes.
Therefore, NSB is evacuation-time-optimal.

Next, we show that NSB is throughput-optimal.
The analysis follows a similar line of argument as in the proof of Theorem~\ref{thm:nsb_throughput} 
for general graphs. We now want to show that for any given arrival rate vector $\lambda$ strictly inside 
$\Lambda^*$, the system is rate stable under the NSB algorithm. Note that $\lambda$ is also strictly 
inside $\Psi$ (i.e., $\lambda_i < 1$ for all $i \in \Vertex$) since $\Lambda^* \subseteq \Psi$.
(In fact, we have $\Lambda^* \subseteq \Psi$ for bipartite graphs.) 
We define $\epsilon \triangleq \min_{i \in \Vertex} (1-\lambda_i)$. Clearly, we must have $\epsilon > 0$. 

Similarly, we proceed the proof using the fluid limit technique. 
Recall that existence of fluid limits (i.e., Eqs.~\eqref{eq:fluid_a}-\eqref{eq:fluid_h}) has been shown 
in Section~\ref{sec:fluid}. We also have the fluid model equations (i.e., Eqs.~\eqref{eq:qqld}-\eqref{eq:diff}).
Then, to show rate stability of the original system, it suffices to show weak stability of the fluid model 
due to Lemma~\ref{lem:weak2rate}. Recall that the Lyapunov function is defined as 
$V(q(t)) = \max_{i \in \Vertex} q_i(t)$.
We want to show that if $V(q(t))>0$ for $t>0$, then $V(q(t))$ has a negative drift. 
Specifically, we want to show that for all regular times $t>0$, if $V(q(t)) > 0$, we have $\diff V(q(t)) \le - \epsilon$.

Recall that $\mC$ is the set of critical nodes in the fluid limits at scaled time $t$ (Eq.~(\ref{eq:critical})),
$\xrj$ denotes a positive subsequence for which the convergence to the fluid limit holds,
and $T \triangleq \{\lceil \xrj t \rceil, \lceil \xrj t \rceil + 1, \dots, \lfloor \xrj (t+\delta) \rfloor\}$ denotes 
a set of consecutive time-slots in the original system corresponding to the scaled time interval $(t,t+\delta)$ 
in the fluid limits, where $\delta$ is a small enough positive number. Now, if we can show that under 
the NSB algorithm, all the nodes in $\mC$ will be scheduled in every time-slot of interval $T$, i.e.,
for all $i \in \mC$, we have $\sum_{M \in \Matching} \sum_{l \in L(i)} M_l \cdot (H_M(\xrj (t+\delta))-H_M(\xrj t)) 
= \lfloor \xrj (t+\delta) \rfloor - \lceil \xrj t \rceil + 1$, then we can show $\sum_{M \in \Matching} \sum_{l \in L(i)} M_l \cdot \diff h_M(t) = 1$ (similar to Eq.~\eqref{eq:dh}),
and thus, it follows from Eq. (\ref{eq:diff}) that for all $i \in \mC$, we have $\diff q_i(t) \le \lambda_i - 1 \le -\epsilon$.

We know from the proof of Lemma~\ref{lem:2in3} that all the nodes in $\mC$ are heavy nodes in 
every time-slot of $T$. Then, following our earlier proof of evacuation time optimality and using 
Lemmas~\ref{lem:existence} and \ref{lem:path}, we can show that all the nodes in $\mC$ will be 
scheduled in every time-slot of interval $T$. This completes the proof of Theorem~\ref{thm:nsb_opt}.
\end{proof}

\section{Proof of Theorem~\ref{thm:lc_nsb}} \label{sec:thm:lc_nsb}
\begin{proof}
The proof is almost the same as that of Theorems~\ref{thm:nsb_et}, \ref{thm:nsb_throughput}, 
and \ref{thm:nsb_opt} for the NSB algorithm.
Hence, in the following we mainly focus on explaining the key differences of the proof and omit the details.

We first want to show that the LC-NSB algorithm has an approximation ratio no greater than $3/2$ 
for the evacuation time. From the proof of Theorem~\ref{thm:nsb_et}, it is not difficult to see that 
the result follows exactly if Proposition~\ref{pro:lc_nsb} (similar to Proposition~\ref{pro:nsb} 
for NSB) holds.

\begin{proposition}
\label{pro:lc_nsb}
Consider any frame.
Suppose the maximum node queue length is no smaller than two at the beginning of a frame. 
Under the LC-NSB algorithm, the maximum node queue length decreases by at least two 
by the end of the frame.
\end{proposition}

The proof of Proposition~\ref{pro:lc_nsb} is the same as that of Proposition~\ref{pro:nsb},
except that we need to replace both $2\Delta(p)$ and $2(\Delta(p)-1)$ in the proof 
of Proposition~\ref{pro:nsb} with 5. This is due to the new way of assigning the 
node weights in Eq.~\eqref{eq:weight_lc} under the LC-NSB algorithm.

Next, we want to show that the LC-NSB algorithm has an efficiency ratio no smaller than $2/3$ for the throughput.
From the proof of Theorem~\ref{thm:nsb_throughput}, it is not difficult to see that the result follows exactly 
if Lemma~\ref{lem:2in3_lc} (similar to Lemma~\ref{lem:2in3} for NSB) holds.
Recall that $\mC$ is the set of critical nodes in the fluid limits at scaled time $t$ (Eq.~(\ref{eq:critical})).

\begin{lemma}
\label{lem:2in3_lc}
Under the LC-NSB algorithm, all the nodes in $\mC$ will be scheduled at least twice within 
\high{each frame of interval $T$}.
\end{lemma}

The proof of Lemma~\ref{lem:2in3_lc} is the same as that of Lemma~\ref{lem:2in3},
except that we need to replace both $2\xrj(q_{\max}(t)-\beta)$ and $\max \{2\xrj(q_{\max}(t)-2\beta), \xrj(q_{\max}(t)+\beta)\}$ 
in the proof of Lemma~\ref{lem:2in3} with 3. This is due to the new way of assigning the 
node weights in Eq.~\eqref{eq:weight_lc} under the LC-NSB algorithm.
 
Finally, we want to show that the LC-NSB algorithm is both throughput-optimal 
and evacuation-time-optimal in bipartite graphs.
The proof of this part is the same as that of Theorem~\ref{thm:nsb_opt},
except that we need to replace $(n-1)/n \cdot \Delta(k)$ in the proof 
of Theorem~\ref{thm:nsb_opt} with 2. This is because the weight of any heavy 
node is at least 2 and the weight of any non-heavy node is equal to 1 in all time-slots,
due to the new way of assigning the node weights in Eq.~\eqref{eq:weight_lc} under the LC-NSB algorithm. 
\end{proof}

\end{document}